\documentclass[superscriptaddress,amsmath,amssymb,twocolumn]{revtex4-1}

\usepackage{MnSymbol}
\usepackage{complexity}
\usepackage{graphicx}
\usepackage{amsmath}
\usepackage{fullpage}
\usepackage{amsfonts}
\usepackage{amssymb}
\usepackage{dsfont}
\usepackage{amsthm}
\usepackage{tensor}
\usepackage{enumitem}
\usepackage{amsthm}
\usepackage{xcolor}
\usepackage{comment}
\usepackage{mathtools}
\usepackage{physics}
\usepackage{bbold}

\usepackage{appendix}
\usepackage[ruled,vlined]{algorithm2e}
\usepackage[colorlinks]{hyperref}
\hypersetup{
pdfstartview={FitH},
pdfnewwindow=true,
colorlinks=true,
linkcolor=blue,
citecolor=blue,
filecolor=blue,
urlcolor=blue}
\usepackage[capitalize]{cleveref}
\usepackage{enumitem}

\newtheorem*{problem*}{Problem}

\def\bH{{\bf H}}
\def\bA{{\bf A}}
\def\bS{{\bf S}}
\def\bX{{\bf X}}
\def\bB{{\bf B}}

\def\bQ{{\bf Q}}
\def\bU{{\bf U}}
\def\bP{{\bf P}}

\def\rT{{\rm T}}

\def\ri{{\rm i}}

\def\rd{{\rm d}}
\def\rR{{\rm R}}
\def\rI{{\rm I}}

\def\cO{\mathcal{O}}

\def\bC{\mathbf{C}}

\def\poly{\mathrm{poly}}

\def\one{{\mathchoice {\rm 1\mskip-4mu l} {\rm 1\mskip-4mu l} {\rm
1\mskip-4.5mu l} {\rm 1\mskip-5mu l}}}
\def\dqc1{\textsc{DQC1}}
\newcommand{\dket}[1]{|#1 \rrangle}

\makeatletter
\newtheorem*{rep@theorem}{\rep@title}
\newcommand{\newreptheorem}[2]{%
\newenvironment{rep#1}[1]{%
 \def\rep@title{#2 \ref{##1}}%
 \begin{rep@theorem}}%
 {\end{rep@theorem}}}
\makeatother

\theoremstyle{mytheorem}

\newtheorem{theorem}{Theorem}
\newreptheorem{theorem}{Theorem}
\newtheorem{lemma}{Lemma}
\newreptheorem{lemma}{Lemma}

\newreptheorem{claim}{Claim}

\newreptheorem{definition}{Definition}

\newreptheorem{remark}{Remark}

\newreptheorem{corollary}{Corollary}

\newreptheorem{observation}{Observation}

\begin{document}

%%%%%%%%%%%%%%%%%%%%%%%%%%%%%%%%%%%%%%%%%%%%%%%%%
\title{Quantum algorithm for linear matrix equations}

\author{Rolando D. Somma}
\email{rsomma@google.com}
\affiliation{Google Quantum AI, Venice, CA 90291, United States}

\author{Guang Hao Low}
\affiliation{Google Quantum AI, Venice, CA 90291, United States}

\author{Dominic W. Berry}
\affiliation{School of Mathematical and Physical Sciences, Macquarie University, Sydney, NSW 2109, Australia}

\author{Ryan Babbush}
\affiliation{Google Quantum AI, Venice, CA 90291, United States}

\date{\today}

\begin{abstract}
We describe an efficient quantum algorithm for solving the linear matrix equation 
${\bf A} {\bf X} + {\bf X} {\bf B} = {\bf C}$, where ${\bf A}$, ${\bf B}$ and ${\bf C}$ are given complex matrices and $\bX$ is unknown. 
This is known as the Sylvester equation, a fundamental equation with applications in control theory and physics. Our approach constructs the solution  
matrix ${\bf X}/x$ in a block-encoding, where $x$ is a rescaling factor needed for normalization. This allows us 
to obtain certain properties of the entries of ${\bf X}$ exponentially faster than would be possible from preparing ${\bf X}$ as a quantum state.
 The query and gate complexities of the quantum circuit that implements this block-encoding are almost linear in a condition number that depends on $\bA$ and $\bB$, and depend logarithmically in the dimension and inverse error.
 We show how our quantum circuits can solve \BQP-complete problems efficiently, discuss potential applications and extensions of our  approach, its connection to Riccati equation, and comment on open problems.
\end{abstract}
\maketitle

%%%%%%%%%%%%%%%%%%%%%%%%%%%%%
\section{Introduction}

Quantum algorithms for linear algebra 
problems have the potential to unlock 
novel applications of quantum computers.
Among the best-known examples is the 
algorithm of Harrow, Hassidim, and
Lloyd 
(HHL)~\cite{HHL09}, which solves the linear system ${\bf A}\vec x= \vec b$ by encoding the solution in a unit quantum state $\ket x \propto \vec x$. Under some assumptions, the complexity of the HHL algorithm or its variants~\cite{CKS17,SSO19,costa2022optimal} is only polylogarithmic in the dimension of ${\bf A}$, thereby providing an exponential quantum speedup in some instances.
Since systems of linear equations are ubiquitous in science, prior works consider potential applications of this algorithm, including data fitting~\cite{wiebe2012quantum}, estimation of hitting times of Markov chains~\cite{CS16}, differential equations~\cite{CJS13,Ber14}, quantum machine learning~\cite{lloyd2013quantum}, and beyond~\cite{liu2022survey}. Nevertheless, identifying
practical applications for the HHL algorithm that yield substantial quantum speedups has proven challenging, due to the stringent requirements for the quantum algorithm to be efficient in applications~\cite{Aar15}. Exploring additional quantum approaches to linear algebra problems is then paramount
to addressing the value proposition of quantum computing in this domain.

This article also considers  
a quantum algorithm for a linear algebra problem, 
namely solving linear matrix equations. Such equations also have a broad range of practical applications in science and engineering, including control theory~\cite{castelan2005solution}, signal analysis~\cite{wei2015fast}, linear algebra~\cite{horn2012matrix}, differential equations~\cite{sorensen2003direct}, and physics~\cite{xu2014sylvester}. 
In contrast to systems of linear equations, the solution to linear matrix equations is an unknown matrix, ${\bf X}$.  
While conceptually matrix equations can also be interpreted as linear systems of large dimension (using vectorization), 
they are almost never solved in this way due to the large dimensions of the matrices. Instead, specialized classical algorithms aim at constructing the matrix ${\bf X}$ using more clever approaches~\cite{bartels1972algorithm,golub1979hessenberg}. A similar observation also applies to our quantum algorithm.

Recently, Ref.~\cite{Liu2025} considered a (non-linear) matrix equation where the goal is to construct a unitary quantum circuit that encodes the solution ${\bf X}/x$
in one of its blocks, i.e., a `block-encoding'.
Here, $x>0$ is needed for normalization.
Reference~\cite{clayton2024differentiable} also
provides a block-encoding solution to the Lyapunov equation, which is a special case of a linear matrix equation.
Block-encodings give a natural model for matrix equations that
contrasts other algorithms for quantum linear algebra, like the HHL algorithm, where the output is a quantum state. For problems related to computing entries of $\bX$, we show that this block-encoding approach can be more powerful: quantum states need to be normalized and the normalization factor is often exponentially small, while the corresponding normalization factor $1/x$ for block-encodings might not scale that way.

We focus on a fundamental matrix equation referred to as the Sylvester equation: ${\bf A} {\bf X} + {\bf X} {\bf B} = {\bf C}$, where ${\bf A}$, ${\bf B}$ and ${\bf C}$ are given. 
For the Sylvester equation, we construct a quantum circuit that prepares a block-encoding of ${\bf X} \in \mathbb C^{N \times N}$ in time that is almost linear (and sublinear in one case) in a condition number $\kappa$ that depends on the matrices $\bA$ and $\bB$, polylogarithmic in the inverse of a precision parameter, and polylogarithmic in $N$, under broad assumptions on the matrices. The block-encoding
is a unitary that contains ${\bf X}/x$ in its first diagonal block, where $x$ roughly scales as $\kappa \alpha$ and is needed for normalization (here, $\alpha$ is the block-encoding normalization of $\bC$, for instance of `$L_1$-norm' of coefficients in its representation as a linear-combination-of-unitaries). In instances where $\kappa$ depends polylogarithmically on $N$, the quantum circuit is efficient and can 
be used to solve problems with an exponential quantum speedup.
Indeed, the quantum circuits can be used to solve $\BQP$-complete problems efficiently, by adapting the $\BQP$-completeness result of Ref.~\cite{HHL09} to this setting, as solving systems of linear equations is an example of our more general problem.
Our results contrast prior work like Ref.~\cite{Liu2025}, which considered the algebraic Riccati equation instead, and required very different methods of solution.

While we do not analyze explicit gate counts in real-world problems, we summarize some potential practical applications of our quantum algorithm
to various science and domains. 
Like in other quantum algorithms for linear algebra such as HHL, some assumptions are needed to obtain an end-to-end algorithm that is efficient (e.g., a polylogarithmic condition number)~\cite{Aar15}. Nevertheless, since we work with block-encodings, the assumptions for the current setting are different and might be less stringent for some problems due to, for example, the different normalization factors that appear in block-encodings versus quantum states.
To this end, we also present a number of exponential separations in query complexity 
when comparing the block-encoding and state access models. These separations are related to computing matrix entries.
We note that our quantum circuits might not be the most efficient, or even most broadly applicable to solving classes of matrix equations. Finding tight lower bounds and quantum approaches for more general instances of matrix equations remain as open problems.

%%%%%%%%%%%%%%%%%%%%%%%%%%%%%%%%%%%%
\section{Quantum linear matrix equation problem}
\label{sec:intro}

We focus on the Sylvester equation
but other linear matrix equations might be considered~\cite{lancaster1970explicit}.
Given matrices ${\bf A} \in \mathbb C^{M \times M}$, ${\bf B} \in \mathbb C^{N \times N}$, and ${\bf C} \in \mathbb C^{M \times N}$, the Sylvester equation is
\begin{align}
\label{eq:sylvestereq}
    {\bf A} {\bf X}+ {\bf X} {\bf B}= {\bC } \;.
\end{align}
The goal is to solve for the unknown matrix ${\bf X} \in \mathbb C^{M \times N}$. 
The matrices are not necessarily Hermitian but,
without loss of generality, we can assume $N=M$
and $\|{\bf A}\|\le 1/2$ and $\|{\bf B}\|\le 1/2$, where $\|.\|$ is the spectral norm; see Appendix~\ref{app:squarematrix}.
Importantly, we require ${\bf A}$ and ${\bf B}$ to be diagonalizable.

In our `quantum version' of this problem the main goal is to prepare a block-encoding of ${\bf X}$. 
This is a quantum circuit where the first block of the unitary matrix representation is proportional to $\bX$.
Since in principle $\|{\bf X}\|$
 can be larger than 1, this will introduce a normalization factor as block-encodings  have unit spectral norm. Formally, we define the quantum linear matrix equation  problem (QLME) as follows.

\begin{problem*}[QLME]
\label{prob:qlme}
    Let ${\bf A}$, ${\bf B}$, ${\bC}$
    be complex matrices of dimension $N\times N$
     satisfying $\|{\bf A}\| \le 1/2$, $\|{\bf B}\| \le 1/2$, $\|{\bf C}\| \le \alpha$, 
     and assume access to their block-encodings $U_{\bf A}$, $U_{\bf B}$, and $U_{\bC/\alpha}$, together with their inverses and controlled versions. Let $\kappa:=\|\bQ^{-1}\|$, where
     ${\bf Q}:={\bf A}\otimes \one_N + \one_N \otimes {\bf B}^\rT$  is assumed to be diagonalizable. Consider the Sylvester equation
     \begin{align}
         {\bf A}{\bf X}+{\bf X}{\bf B}={\bf C} \;.
     \end{align}
     The goal is to output a classical description of a quantum circuit $U_{{\bf X}/x} \in \mathbb C^{M \times M}$, for   $x \ge \kappa \alpha$ and   $M=\cO(\poly(N))$, that satisfies
     \begin{align}
         \|\Pi U_{{\bf X}/x} \Pi - {\bf X}/x \| \le \epsilon \;,
     \end{align}
     where  and $\Pi$ is the $N$-dimensional orthogonal projector and $\epsilon$ is the error. 
     That is, $U_{{\bf X}/x}$ is an approximate block-encoding of the normalized solution ${\bf X}/x$.
\end{problem*}

In our notation, $\one_N$ is the $N$-dimensional identity and $\bB^\rT$ is the transpose of $\bB$.
We note that $\|\bX\| \le \kappa \alpha$ and, since this bound is tight, we need the rescaling factor $x$ in the block-encoding. 
Indeed, the closer $x$ is to $\kappa \alpha$ the better, and we are able to achieve this up to logarithmic corrections.
Similar to quantum simulation problems where the goal is to provide a quantum circuit that simulates the evolution operator,
the QLME considers a quantum circuit that applies $\bX/x$. While technically 
this problem is solved classically (i.e., the QLME is a classical description of the quantum circuit), the end goal is to use 
$U_{\bX/x}$ in a quantum algorithm that `solves' the Sylvester equation.
For example, access to $U_{\bX/x}$ allows for the preparation of (column) states $\bX \ket j/\|\bX \ket j\|$, by simply preparing $U_{\bX/x}\ket j \ket 0$ and using amplitude amplification. It also allows for the computation of the entries $\bra j \bX \ket k$ by computing the corresponding expectation 
$\bra j \bra 0 U_{\bX/x}\ket k \ket 0$.

Our main result is an efficient quantum circuit that solves the QLME in certain cases.
\begin{theorem}
\label{thm:main}
    The quantum circuits that solve the QLME use the following number of queries to the block-encodings of $\bA$ and $\bB$, including their inverses and controlled versions, and number of arbitrary two-qubit gates:
    \begin{enumerate}[leftmargin=*]
        \item Normal matrices. If the matrix $\bQ$
        is normal, the query complexity for $x \propto \kappa \alpha \sqrt{\log(N/ \epsilon)}$ is
        $Q=\cO(\kappa \log( \kappa N/\epsilon))$ and the gate complexity is also $G=\cO(\kappa \log( \kappa N/\epsilon))$.
        \item Positive Hermitian part. If $\bQ_H \succ 0$ is the Hermitian part of $\bQ$, the query complexity for $x \propto \kappa \alpha \log(\kappa \|\bC\|/ \epsilon)$ is
        $Q=\cO(\kappa \polylog( \kappa \|\bC\|/\epsilon))$ and the gate complexity is also $G=\cO(\kappa \polylog( \kappa \|\bC\|/\epsilon))$.
        \item $\bB=0$. The query complexity for $x \propto \kappa \alpha \sqrt{\log(1/\epsilon)}$ is
        $Q=\cO(\kappa \log(1/\epsilon))$ and the gate complexity is also $G=\cO(\kappa \log (1/\epsilon))$.
    \item
    Positive matrices.
    If $\bQ \succ 0$ is positive where $\bA=(\bP_\bA)^2$ and $\bB=(\bP_\bB)^2$,
    assuming access to block-encodings of $\bP_\bA$ and $\bP_\bB$,
    we obtain improved results. The query complexity for $x \propto \kappa \alpha \log(  N/\epsilon)$ is $Q=\cO(\sqrt \kappa \log(\kappa N/\epsilon))$ and the gate complexity is also $G=\cO(\sqrt \kappa \log(\kappa N/\epsilon))$.
    \end{enumerate}
\end{theorem}
The asymptotic parameters are $1/\epsilon$, $\kappa$, and $N$. The circuit $U_{\bX/x}$ also requires a single call to $U_{\bC/\alpha}$.

To prove Thm.~\ref{thm:main} we
 use a standard technique
based on linear combinations of unitaries (LCU)~\cite{CKS17,BCC+15,LC19}. Such unitaries involve Hamiltonian evolutions with Hermitian matrices obtained from ${\bf A}$ and ${\bB}$. 
These evolutions can be simulated using the respective block-encodings via quantum signal processing~\cite{LC17}: simulating $e^{-\ri t \bH}$
for Hermitian $\bH$ and error $\epsilon$ requires $\cO(t +\log(1/\epsilon))$ queries to the block-encoding $U_\bH$ of $\bH$ plus arbitrary gates.
The query complexities for the QLME are then determined by the largest evolution times and the allowable error. Additional gates are needed for the implementation of the LCU using ancillary registers. The detailed proofs for each case are in Appendix~\ref{app:complexity}. 

Our approach to constructing $U_{\bX/x}$ is related to that  in Ref.~\cite{CKS17} for the quantum linear systems problem that uses a linear combination of Hamiltonian simulations. However,
we will see that while ${\bf X}$ is   linear in ${\bC}$, this can involve exponentials of both ${\bf A}$ and ${\bf B}$. For this reason, Thm.~\ref{thm:main} is solved using expressions for the inverse operator that only apply to the specific cases discussed. It is open to provide a general solution that gives an efficient quantum circuit for the QLME for all invertible ${\bf Q}$.

%%%%%%%%%%%%%%%%%%%%%%%%%%%%%%%%
\section{Solutions to the Sylvester equation}

We describe some identities needed for our main result.
There is vast literature on solutions   useful for {\em classical} algorithms~(Cf.\ \cite{lancaster1970explicit,bartels1972algorithm,sorensen2003direct}). Our setting is  distinct, since we require expressions for solutions that  are amenable to {\em quantum} algorithms, given as LCUs.
These solutions
are better understood  using `vectorization'.
For a matrix ${\bf F} =\sum_{j,k} f_{jk} \ket j \bra k$, we let $\dket{F}:=\sum_{j,k} f_{jk}\ket{j,k}$ be the vectorized form of ${\bf F}$. For the Sylvester equation, we obtain
\begin{align}
  ( {\bf A} \otimes \one_N)  \dket{X} + (\one_N \otimes  {\bf B}^\rT)  \dket{X}= \bQ \dket{X}=\dket{C} \;.
\end{align}
This is a system of linear equations satisfying
\begin{align}
\label{eq:vectorsolution}
   \dket{X} =\frac 1 {\bf Q} \dket{C} \;.
\end{align}
If $\gamma_j$ and $\lambda_k$ denote eigenvalues of ${\bf A}$ and ${\bf B}$, respectively, the solution  is unique as long as ${\bf Q}$ is invertible, i.e., when $\gamma_j + \lambda_k \ne 0$ for all $j,k \in \{1,\ldots,N\}$. We let $\kappa = \|\bQ^{-1}\|$, which is related to
 the condition number of ${\bf Q}$.

The vectorization approach would readily
provide a quantum algorithm for the Sylvester equation
using the HHL algorithm.
However, here we are interested in obtaining a block-encoding of $\bX/x$ and we would
like to express the solution in a different form that is useful to our approach. 
To this end, we consider integral identities for the inverse, which  
require some assumptions on the matrices $\bA$ and $\bB$.  

1. {\bf Normal matrices.}  
        This case assumes $[{\bf Q},{\bf Q}^\dagger]=0$ 
        and we decompose ${\bf Q}={\bf Q}_H + \ri {\bf Q}_S$,
        where ${\bf Q}_H$ and ${\bf Q}_S$ are Hermitian, and $[{\bf Q}_H,{\bf Q}_S]=0$. 
        Since $\bQ \bQ^\dagger = \bQ_H^2+\bQ_S^2 \succ 0$, we consider the identity
        \begin{align}
            \frac 1 {\bf Q}
            % & = \int_{0}^\infty \rd t \;  (t {\bf Q}^\dagger ) e^{-t^2 {\bf Q}{\bf Q}^\dagger /2} \\
            &= \int_{0}^\infty \rd t \;  (t {\bf Q}^\dagger) e^{-t^2 {\bf Q}_H^2/2}
            e^{-t^2 {\bf Q}_S^2/2} \;.
        \end{align}
        A Fourier transform gives
        \begin{align}
       ( t  \bQ_H )  e^{-t^2 {\bf Q}_H^2/2} = \frac{\ri}{\sqrt{2\pi}} \int_{-\infty}^\infty \! \! \rd \omega \;\omega e^{-\omega^2/2}e^{-\ri \omega t {\bf Q}_H} \;.
        \end{align}
        Using a similar expression for $(t \bQ_S) e^{-t^2 {\bf Q}_S^2}$ allows us to express $\frac 1 {\bf Q}$ as a triple integral involving exponential operators of the form
        $e^{-\ri \omega t {\bQ_H} -\ri \omega' t {\bQ_S}}$,
        which can be expressed using ${\bf A}$
        and $\bB$.
        Applying the inverse of vectorization to Eq.~\eqref{eq:vectorsolution}
        gives
        \begin{align}
        \nonumber
            {\bf X}&= \frac 1 {2\pi}\int_{0}^\infty \! \! \rd t \int_{-\infty}^\infty \! \! \rd\omega \int_{-\infty}^\infty \! \! \rd\omega' \; e^{-(\omega^2+\omega'^2)/2} (\ri\omega+ \omega')\times \\
            \label{eq:case1solution}
            & \quad \quad \times
            e^{-\ri \omega t {\bf A}_H-\ri \omega' t {\bf A}_S} {\bC} e^{-\ri \omega t {\bf B}_H-\ri \omega' t {\bf B}_S} \;.
        \end{align}
%%%%%%

      2.  {\bf Positive Hermitian part.}  This assumes ${\bf Q}_H \succ 0$ and
         $\Re(\gamma_j+\lambda_k)>0$.
       We consider the identity
        \begin{align}
        \label{eq:inversewithexponential}
            \frac 1 {\bf Q} = \int_0^\infty \rd t \; e^{-t {\bf Q}} \;.
        \end{align}
        Let $f(\omega)$ be such that
        $e^{-x}=\int_{-\infty}^{\infty} d \omega \; f(\omega) e^{-\ri \omega x}$ for $x \ge0$ and unspecified otherwise.  Then, Ref.~\cite{an2023quantum} gives
        \begin{align}
           e^{-t {\bf Q}} = \int_{-\infty}^{\infty} d \omega \; f(\omega)  e^{-\ri t ( \omega  {\bf Q}_H + {\bf Q}_S)} \;,
        \end{align}
        for many such functions.
 Applying the inverse of vectorization to Eq.~\eqref{eq:vectorsolution}
        gives  
        \begin{align}
         \label{eq:case2solution}
            {\bf X}=  \int_0^\infty \! \! \! \rd t \int_{-\infty}^{\infty} \! \! \! \rd \omega \; f(\omega) 
            e^{-\ri t ( \omega  {\bf A}_H - {\bf A}_S)} \bC e^{-\ri t( \omega  {\bf B}_H - {\bf B}_S)} \;.
        \end{align}
        %%%%%%

3. {\bf $\quad {\bf B}={\bf 0}$}.
        This case is similar to the quantum linear systems problem~\cite{CKS17} and we can directly compute $1/\bA$
        and apply it to $\bC$ without using vectorization. The matrix $\bA$ is Hermitian without loss of generality in this case.
        We use the identity
        \begin{align}
            \frac 1 {\bf A} = \int_0^{\infty} \rd t \; (t {\bf A}) e^{-t^2 {\bf A}^2/2} \;,
        \end{align}
    and the Fourier transform to express
        \begin{align}
            (t {\bf A}) e^{-t^2 {\bf A}^2/2} = \frac{\ri}{\sqrt{2\pi}}\int_{-\infty}^\infty \rd \omega \; \omega e^{-\omega^2/2}e^{-\ri \omega t {\bf A}} \;.
        \end{align}
         The solution to Eq.~\eqref{eq:sylvestereq}
        is the double integral
        \begin{align}
         \label{eq:case3solution}
            \bX =\int_0^{\infty} \rd t \int_{-\infty}^\infty \rd \omega \; \omega e^{-\omega^2/2} e^{-\ri \omega t {\bf A}} \bC \;.
        \end{align}
        (The case 
        ${\bf A}={\bf 0}$ is similar, only that we need to act with $e^{-\ri \omega t {\bf B}}$ on the right of $\bC$.)

%%%%%%%

        4. {\bf Positive matrices.} This case 
         assumes access to matrices ${\bP}_\bA$ and 
        ${\bP}_\bB$ such that
        ${\bf A}=({\bP}_\bA)^2 \succ 0$ and 
        ${\bf B}=({\bP}_\bB)^2 \succ 0$,  implying a faster algorithm. Then $ e^{-t {\bf Q}}= e^{-t ({\bP_\bA})^2} \otimes e^{-t ({\bP_\bB^\rT})^2}$ in Eq.~\eqref{eq:inversewithexponential} and the Hubbard-Stratonovich identity~\cite{hubbard1959calculation} implies
        \begin{align}
        \nonumber
            {\bf X}=\frac 1 {2\pi} \int_0^\infty & \rd t \int_{-\infty}^\infty \rd \omega
            \int_{-\infty}^\infty \rd \omega' \; e^{-(\omega^2 + \omega'^2)/2} \times \\
            \label{eq:case4solution}
           & \times e^{-\ri \omega \sqrt{2t} {\bP}_\bA} \bC e^{-\ri \omega' \sqrt{2t} {\bP}_\bB} \;.
        \end{align}
%%%%%%

Unfortunately, obtaining useful expressions for the general case
seems challenging. For example, it is not obvious 
to find a suitable reduction from the general problem to one where ${\bf Q}$
is normal or Hermitian, which contrasts the case of
the quantum linear systems problem where ${\bf A}$ (or ${\bf B}$) can be always assumed to be Hermitian.
One method of generalization is to approximate the inverse of ${\bf Q}$ by a polynomial, but that approach yields exponential complexity in $\kappa$ (see Appendix \ref{app:nonpos}).
Also, note that we can multiply all matrices by the same phase, so all prior assumptions on $\bA$ and $\bB$
are up to a phase.

%%%%%%%%%%%%%%%%%%%%%%%%%%%%%%%%%%
\section{Quantum circuits for the QLME}

The prior expressions show how $\bX/x$ can be constructed using time evolutions with respective matrices. Such expressions readily give an idea of the query complexities in Thm.~\ref{thm:main},
 essentially determined by the largest evolution times. While the integrals are up to $t \rightarrow \infty$, it suffices to  cut them off at some time $T$ that is linear on $\kappa$ and (sub)logarithmic in the inverse precision; e.g., $\|e^{-T \bQ}\|=e^{-T/\kappa}=\cO(\varepsilon)$ for $T =\Omega(\kappa \log(1/\varepsilon))$. It is also possible to approximate the integrals over $\omega$ and $\omega'$ by placing proper cutoffs since the integrals involve Gaussians.
Then,
to construct $U_{\bX/x}$,
we approximate the  integrals by finite sums.   These are discussed in Appendix~\ref{app:complexity} in great detail.
The result is an approximation of $\bX$ as an LCU, and we use standard approaches to implement it, like the following well-known result.
\begin{lemma}[Adapted from Lemma 6 of Ref.~\cite{CKS17}]
\label{lem:BEfromLCU}
    Let ${\bf M}=\sum_i \alpha_i U_i$ be an LCU with $\alpha_i>0$ and $U_i$ unitary. Let $V$ be any unitary that maps $V \ket 0 = \frac 1 {\sqrt \alpha }\sum_i \sqrt{\alpha_i}\ket i$, where $\alpha=\sum_i \alpha_i$. Then, if $U:=\sum_i \ketbra i \otimes U_i$, $V^\dagger U V$ is a block-encoding of ${\bf M}/\alpha$:
    \begin{align}
        \Pi V^\dagger U V \Pi = \frac 1 \alpha {\bf M} \;.
    \end{align}
\end{lemma}
Often, $V$ and $U$ are referred to as the PREPARE and SELECT operations~\cite{BCC+15}.
This Lemma and the results in Supp.\ Mat.\ prove Thm.~\ref{thm:main}. Figure~\ref{fig:normalcase}
illustrates $U_{\bX/x}$ for the case of normal matrices.

\begin{figure}[htb]
\includegraphics[scale=.25]{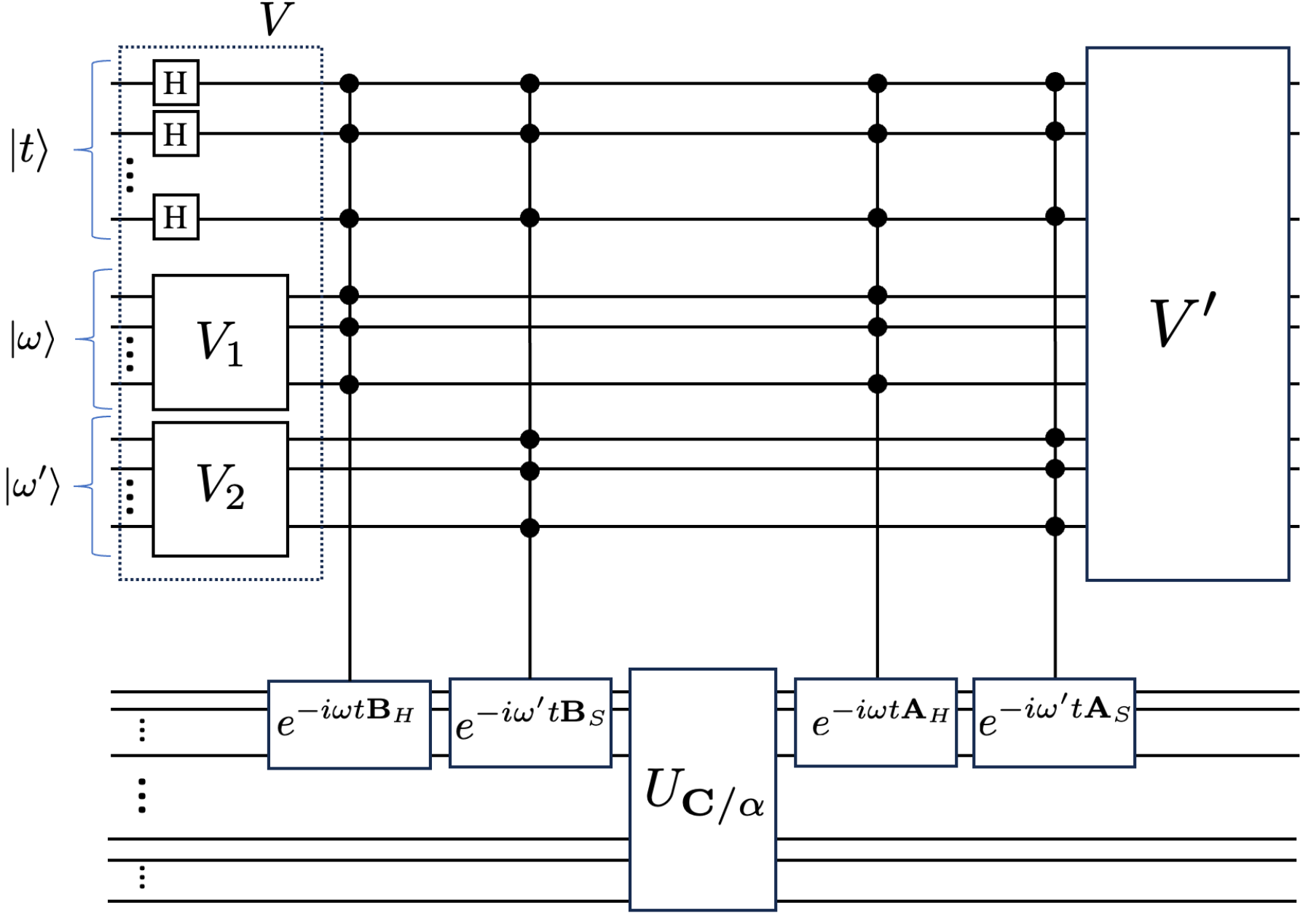} 
\caption{An approximate block-encoding of $\bX/x$ when the matrices are normal. Three ancillary registers are used to encode the discrete values of $t$, $\omega$, and $\omega'$. The circuit gives a block-encoding that approximates the integral in Eq.~\eqref{eq:case1solution} by a finite LCU. Illustrated is only one of the two terms (due to $\omega + \ri \omega'$) in Eq.~\eqref{eq:case1solution}. The black circles denote standard controlled operations. The unitaries $V_1$ and $V_2$ produce states that encode the coefficients in the LCU.  The unitary $V'$ is similar to $V^\dagger$, up to phases in the computational basis.
}
\label{fig:normalcase}
\end{figure}

%%%%%%%%%%%%%%%%%%%%%%%%%%%%%%%%%%%%%%
\section{Applications}

Linear matrix equations are fundamental in   science and we give a 
list of potential applications of our quantum circuits.
Nevertheless, like other quantum algorithms, several requirements must be met for efficiency. Such requirements constrain the set of problems and instances for which quantum advantage might be possible.
Performing detailed complexity analyses and comparing with specialized classical algorithms will be important to characterize those real-world problems that can benefit from our results. 
At the same time we note that the QLME allows
for solving $\BQP$-complete problems efficiently. For example, in an instance where $\bB=0$ and $\bC=\one_N$, the circuit $U_{\bX/x}$
is a block-encoding of $\frac 1 {\kappa \bA}$, and results in Ref.~\cite{HHL09} show that this inverse can be used to simulate any quantum circuit efficiently, highlighting the potential of an exponential quantum advantage.

%%%%%
\vspace{0.1cm}

{\bf Matrix functions.} The Sylvester equation can be used in matrix block-diagonalization~\cite{horn2012matrix,dmytryshyn2015coupled}. Consider the block matrix 
\begin{align}
 \bS:=   \begin{pmatrix} \bA & \bC \cr {\bf 0}& \bB
    \end{pmatrix} =\begin{pmatrix} \one_N & -{\bf X} \cr {\bf 0}& \one_N
    \end{pmatrix} \begin{pmatrix} \bA & {\bf 0} \cr {\bf 0}& \bB
    \end{pmatrix}\begin{pmatrix} \one_N & {\bf X} \cr {\bf 0}& \one_N
    \end{pmatrix} \;.
\end{align}
The diagonalizing transformation is such that
\begin{align}
    \bA \bX - \bX \bB = \bC \;,
\end{align}
which is a Sylvester equation. By producing a block-encoding of $\bX$, and hence that of the diagonalizing transformation, we can construct block-encodings of functions of $\bS$. For example, if $\bA$ and $\bB$ are Hermitian, we can use QSVT to compute a block-encoding of a function $f$ of the block-diagonal matrix, and then use the block-encoding of $\bX$ to compute a block-encoding of $f(\bS)$, which is an alternative approach to quantum eigenvalue processing~\cite{Low2025Eigenvalue}.

\vspace{0.1cm}
%%%%%
{\bf Differential equations.} 
    Examples like the two-dimensional Poisson equation $\partial_v^2 x(v,w) + 
    \partial_w^2 x(v,w) = c(v,w)$, where $c(v,w)$ is the source,  classical algorithms discretize to a grid to write the equation as
    \begin{align}
        {\bf P}^2 \bX + \bX  {\bf P}^2 = \bC \;,
    \end{align}
    where $\bX$ encodes the solution as $\bU \approx \sum_{v,w} x(v,w) \ket v \bra w$, 
    $\bC$ encodes the source as $\bC \approx \sum_{v,w} c(v,w) \ket v \bra w$, and $\bP^2$ is the discretized Laplacian. This example can be addressed with our quantum algorithm. Note that the Laplacian condition number is $\cO(N^2)$,
    but since this is an example of case 4, the complexity of the circuit is $\tilde {\cal O}(N)$. In contrast, exact classical algorithms might suffer from a $\cO(N^2)$ scaling when needing to represent the full $N \times N$ matrix $\bX$.

    Matrix equations also appear within the context of Sylvester flow linear inhomogeneous differential
    equations $\frac {\rd}{\rd t}\bX(t) =\bC - \bA \bX(t) - \bX(t) \bB$. The fixed point of such equation satisfies the Sylvester equation, for which our quantum algorithm can be used.

    \vspace{0.1cm}
    %%%%%%%%%%%%%%%%%%%%%%
{\bf Uncertainty quantification (UQ).}
    For certain differential equations, 
    we can introduce an additional (unknown) parameter $s$ in the input and consider different realizations for different values of $s$. For example, for the Poisson equation, we obtain $\partial_v^2 x(v,w,s) + 
    \partial_w^2 x(v,w,s) = c(v,w,s)$. In UQ one is interested in the response for various $s$. When considering an average over $s$, for example, we are interested in solving the Sylvester equation
     \begin{align}
        {\bf P}^2 \overline \bX + \overline\bX  {\bf P}^2 = \overline\bC \;,
    \end{align}
    where $\overline\bC$ is the average of the source term.

    Beyond, the instances where $\bB=\bf 0$ can be interpreted as representing multiple instances of the quantum linear systems problem for exponentially-many input vectors $\vec C_j$, which are encoded in the columns of a matrix $\bC$. Hence, that case allows for studies of UQ in systems of linear equations.

\vspace{0.1cm}
%%%%%%
{\bf Physics and control theory.}
The Sylvester equation is ubiquitous in these domains. For example, the same differential equation  $\frac {\rd}{\rd t}\bX(t) =\bC - \bA \bX(t) - \bX(t) \bB$ governs the dynamics of correlation matrices under noise or dissipation (open quantum system dynamics), known as the Lyapunov equation when $\bA = \bB^\rT$.
The steady state satisfies the Sylvester equation. This can also be interpreted as a generalization of the results in Ref.~\cite{somma2025shadow} on shadow Hamiltonian simulation to the case where the evolution is not unitary.
Examples are models of quantum or classical oscillators with friction and external noise.

In perturbation theory, the Schrieffer-Wolff transformation is used to remove a perturbation $V$ (at first order) in the Hamiltonian $H=H_0+V$, resulting in the equation $-\ri[S,H_0]=V$. This is also a Sylvester equation, and while the solution $S$ is not unique (e.g., $S'=S+\beta \one$ is also a solution), one often adds a regularizing term  to obtain a proper solution: $-\ri[S,H_0]+\gamma S=V$, for small $\gamma >0$.
Since the smallest eigenvalue of $\bQ$ is now of magnitude $\gamma$ and $H_0$ is Hermitian, the system is well conditioned and we can use the case 1 results to solve for $S$. While  this is a problem that may be addressed classically efficiently for some instances, such classical algorithms do not apply in general.

Within the context of control theory,
an example problem is that of 
minimizing the Frobenius norm $\| {\bf A}{\bf X} + {\bf X}{\bf B}-{\bf C}\|_{\rm Fr}$ where the corresponding system is overdetermined. Again, one approach to solving this would be via vectorization and using the pseudoinverse of $\bQ$. But this approach can be inefficient
if one is interested in certain matrix properties $\bX$, due to the rescaling factor appearing during vectorization. Instead, we can construct a quantum circuit that encodes $\bX$
from the identity for the pseudoinverse of $\bQ$, and obtain the solution to this minimization problem. This generalizes the results in Ref.~\cite{wiebe2012quantum} for finding a least-square fit
to the current setting of linear matrix equations.

\vspace{0.1cm}

One of the reasons we focus on examples in differential equations and physics is that, when Sylvester equations emerge in those contexts, there is often some intrinsic structure on the matrices $\bA$, $\bB$ and $\bC$ that naturally gives a succinct representation for them. For instance, if these matrices are related to the interactions of a physical system (e.g., interactions in a lattice or a Laplacian), they often are sparse and structured in a way that would facilitate efficient block-encodings. Conversely, less structured problems arising in contexts like machine learning may present greater challenges, including the inability of finding
 efficient block-encodings of the data.

We also note that access to $\bX/x$ allows
for computing matrix entries or expectation values of $\bX/x$ efficiently, or even learning $\bX/x$ efficiently, by using standard techniques (e.g., the Hadamard test). 
Hence, the type of applications considered should 
be those that benefit from these features.

%%%%%%%%%%%%%%%%%%%%%%%%%%%%
\section{Oracle separations}
\label{sec:separations}

The QLME differs from the definition
of other quantum linear algebra problems in that the goal is to prepare a block-encoding of the solution and not a quantum state. 
Having access to ${\bf X}/x$ rather than $\dket{X}/\|\dket{X}\|$
allows for certain efficient computations that would otherwise require exponentially many copies of the state  (and vice versa). The following results, proven in Appendix~\ref{app:separations}, highlight these oracle separations.
They do not concern solutions to a specific matrix equation but rather the two different forms of access.
Note that when we discuss the number of block-encodings required, this includes their inverses and controlled versions.

\begin{theorem}[Exponential separation in computing matrix entries]
\label{thm:separation1}
    Let $\bC=\sum_{j,k=1}^N c_{jk}\ket j \bra k \in \mathbb C^{N \times N}$
    be a permutation matrix and assume block-encoding access
    to $\bC$. 
   Let $\dket C = \sum_{j,k=1}^N c_{jk}\ket{j,k}$ be the vectorization of $\bC$ and assume access to a unitary that prepares the corresponding unit state, i.e., $U_\bC\ket 0 = \dket C/\|\dket C\|$.
    Consider the decision problem of determining whether ${\bra 0}\bC\ket 0$ is 0 or 1. Then, this problem can be solved with
    probability 1 with one query to $\bC$ but requires $\Omega(\sqrt N)$ queries to $U_\bC$ to solve it w.h.p..
\end{theorem}

The following result is motivated to show a separation
when solving `multiple' instances of the quantum linear systems problem. In that case, rather than applying $1/\bA$ to some initial vector, we apply it to a matrix that encodes all instances. We show that this model can  be more powerful than the one where each vector can only be queried `individually'.

\begin{theorem}
\label{thm:separation2}
    Let $\bC=\sum_{j,k=1}^N c_{jk}\ket j \bra k \in \mathbb C^{N \times N}$
    be block-encoded by a unitary matrix $U_{\bC/\alpha}$ for constant $\alpha >0$, i.e., $\|\bC\|\le \alpha$. Let $\ket{C_k}=\sum_j c_{jk} \ket j$ be the columns of $\bC$ and assume they have bounded Euclidean norm $\in \{1/\sqrt 2,1\}$. 
    Let $V_\bC$ be the unitary oracle that performs the map
    \begin{align}
       V_\bC \ket {k,0} \mapsto 
     \frac 1 {\|\ket{C_k}\|}  \ket{k,C_k}\;.
    \end{align}
    Then, $V_\bC$ can be implemented with a constant number of queries to $U_{\bC/\alpha}$ and it inverse, while $U_{\bC/\alpha}$ requires $\Omega(\sqrt N)$  
    queries to  $V_\bC$ and its inverse.
\end{theorem}

Last, we show a separation in the other direction. For certain problems, often related to sampling, having access to the state can be more powerful than having access to the block-encoding.

\begin{theorem}
\label{thm:separation3}
    Let $\bC=\sum_{j,k=1}^N c_{jk}\ket j \bra k \in \mathbb C^{N \times N}$
    be a diagonal matrix where all entries are 0 except for one entry being 1, and $U_\bC$ its block-encoding. Let $\dket C = \sum_{j,k=1}^N c_{jk}\ket{j,k}$ be the vectorization of $\bC$ and assume access to a unitary that prepares this unit state, i.e., $V_\bC \ket 0 = \dket C$. Consider the problem of finding the location of the 1 in $\bC$. Then, it requires $\Omega(\sqrt N)$ queries to  $U_\bC$ to solve this problem while a single copy of $\dket C$ readily determines the location.
\end{theorem}

This is the search problem and the quantum lower bound $\Omega(\sqrt N)$ is well-known~\cite{BHMT02}, while a measurement of $\dket C$ readily gives the location.

%%%%%%%%%%%%%%%%%%%%%%%%%%%%%%%%%%%%%%%%
\section{Connection to Riccati equation}

The algebraic Riccati equation
is the non-linear matrix equation
${\bf B}^\rT {\bf X} + {\bf X}{\bf B} = {\bf C} + {\bf X} {\bf D} {\bf X}$,
where the solution $\bX$ and also $\bC$ are symmetric matrices.
When the nonlinear term vanishes and $\bf D=\bf 0$, this is a Sylvester equation. This Riccati equation is generalized
in several ways such as
\begin{equation}
\label{eq:riccati2}
    {\bf B}^\dagger {\bf X} + {\bf X}^\dagger {\bf B} = {\bf C} + {\bf X}^\dagger {\bf D} {\bf X} \; ,
\end{equation}
which is the case considered in 
Ref.~\cite{Liu2025}. Again, if ${\bf D}=\bf 0$
and $\bB$ and $\bC$ are Hermitian,
  then we can solve for $\bB \bX + \bX \bB=\bC$,
  which is also the Sylvester equation for $\bA=\bB$.

More generally, the solution of Eq.~\eqref{eq:riccati2} is given by \cite[Lemma 5]{Liu2025}
\begin{equation}
    {\bf X} = {\bf D}^{-1}{\bf B} \pm {\bf D}^{-1} \# ({\bf B}^\dagger {\bf D}^{-1} {\bf B} - {\bf C}) \, ,
\end{equation}
where $\#$ indicates the matrix geometric mean, which is
\begin{equation}
    {\bf D}^{-1} \# {\bf C} = {\bf D}^{-1/2} ({\bf D}^{1/2} {\bf C} {\bf D}^{1/2})^{1/2} {\bf D}^{-1/2} \, .
\end{equation}
The conditions for this solution are that ${\bf D}$ and $-{\bf C}$ are positive definite, and ${\bf D}^{-1}{\bf B}$ is Hermitian.

The solution in terms of the matrix geometric mean is fundamentally different from the approaches we use to solve the Sylvester equation here, but it can be used to derive the formula for positive matrices in the special case that $\bA=\bB$.
The solution of the algebraic Riccati equation is analogous to the usual solution to scalar quadratic equations, and the solution of the linear equation can similarly be obtained in the limit of small quadratic term.
In the case where ${\bf D}=D\one_N$, for $D \ll 1$, the minus solution simplifies to
\begin{align}
    &{D}^{-1}\left[{\bf B} - ({\bf B}^\dagger {\bf B} - D{\bf C})^{1/2} \right]\nonumber \\
&={D}^{-1}\left[{\bf B} - ({\bf B}^\dagger {\bf B})^{1/2}+D\int_0^\infty dt \, e^{-t{\bf F}}{\bf C}e^{-t{\bf F}} \right]\nonumber \\
& \quad + \mathcal{O}(D) \, ,
\end{align}
where ${\bf F}=({\bf B}^\dagger {\bf B})^{1/2}$ and the integral expression is obtained from the Taylor expansion of the square root matrix function \cite{Moral2018}.
For positive definite ${\bf B}$ and taking the limit $D \rightarrow 0$, this gives the solution
\begin{equation}
    {\bf X} = \int_0^\infty \rd t \, e^{-t{\bf B}}{\bf C}e^{-t{\bf B}} \, .
\end{equation}
This case is equivalent to the Lyapunov equation, and a similar integral form was used in Ref.~\cite{clayton2024differentiable}.
The Sylvester equation is more general, and for $\bA$ and $\bB$ with positive Hermitian part the solution is of the form
\begin{align}
    \bX = \int_0^\infty \rd t \; e^{-t{\bA}}\bC e^{-t{\bB}} \, .
\end{align}

An even more general form is the nonsymmetric algebraic Riccati equation (NARE) \cite{Bini2010}
\begin{equation}
    {\bf A} {\bf X} + {\bf X} {\bf B} = {\bf C} + {\bf X} {\bf D} {\bf X} \, ,
\end{equation}
which includes the Sylvester equation as a special case.
Unlike the algebraic Riccati equation, there is no known closed-form solution of the NARE.
In particular, there are no known solutions of the NARE or Sylvester equation in terms of the matrix geometric mean.
This means that there does not appear to be a way of applying the approach of Ref.~\cite{Liu2025} to the Sylvester equation studied here.

%%%%%%%%%%%%%%%%%%%%%%%%%%%%%%%%%%%%%%%%
\section{Discussion}

We presented efficient quantum circuits for solving linear matrix equations of the Sylvester type,  expanding the range of applications of quantum computers to linear algebra problems.
Similar to Refs.~\cite{Liu2025,clayton2024differentiable}, but in contrast to most prior work (e.g., HHL) whose goal is to prepare a quantum state, the goal of the QLME problem is to give a classical, efficient description of a unitary block-encoding of the solution matrix. This block-encoding model enables certain computations, such as estimating matrix entries, exponentially faster than the state access model as we showed.

Several open questions remain. One involves identifying specific applications and conducting thorough end-to-end analyses of the algorithm to demonstrate quantum advantage in real-world problems. Although the circuits can be used to efficiently solve BQP-complete problems, with the potential of an exponential quantum speedup, this pertains to worst-case instances that may not often occur in practical settings. For example, while classical algorithms take time polynomial in $N$ to solve the Sylvester equation, such algorithms might be improved in examples where the matrices are well-conditioned and model systems with geometrically local interactions. Other open problems include relaxing the assumptions on the matrices to provide a more general solution, or considering generalizations of the quantum algorithm to more general matrix equations, or even generalizing the theory of QSVT~\cite{GSYW18} to this domain. 
Furthermore, we do not expect the presented algorithm to be query or gate optimal for some tasks, like state preparation (e.g., preparing $\bX \ket i/\|\bX \ket i\|$), where techniques like variable time amplitude amplification~\cite{Amb12} or adiabatic evolution~\cite{SSO19} could be useful. Constructing improved algorithms and proving tight lower bounds remain as relevant problems.

\section{Acknowledgements}

We thank Robin Kothari for discussions.
DWB worked on this project under a sponsored research agreement with Google Quantum AI.
DWB is supported by Australian Research Council Discovery Project DP220101602.

%%%%%%%%%%%%%%%%%%%%%%%%%%%%%%%%%%%%%%%
%\bibliographystyle{apsrev4-2}
%\bibliography{rolando_ref.bib}
%

\newpage

\onecolumngrid

\appendix

%%%%
\section{From rectangular to square matrices}
\label{app:squarematrix}

In general, the solution to 
Sylvester equation is a potentially rectangular matrix  
${\bf X}\in \mathbb C^{M \times N}$, with $M \ne N$.
We present a simple reduction of this case to a Sylvester equation where the solution is a square matrix.
For simplicity we assume that initially $M \le N$.
(The analysis for the  case $M \ge N$ is similar.)

Without loss of generality, $\|{\bf A}\|\le 1/2$
and $\|{\bf B}\|\le 1/2$. Then, $\gamma_j + \lambda_k \in [-1,1]$. We can choose a constant $c>1$ and define
\begin{align}
    {\bf A}':= \begin{pmatrix}
        {\bf A } & {\bf 0} \cr {\bf 0} & c \one_{N-M}
    \end{pmatrix} \;,
\end{align}
where ${\bf 0}$ is an all-zero matrix whose dimension is clear from context. Note that ${\bf A}' \in \mathbb C^{N \times N}$ and $\bA'$ is diagonalizable is $\bA$ is. Consider now the modified Sylvester equation
\begin{align}
\label{eq:appmodifiedSylvester}
   {\bf A}' {\bf X}'  + {\bf X}' {\bf B} = \bC' = \begin{pmatrix}
       \bC \cr {\bf 0} 
   \end{pmatrix}\;,
\end{align}
where $\bC \in \mathbb C^{M \times N}$, $\bC' \in \mathbb C^{N \times N}$, and the solution is also a square matrix ${\bf X}' \in \mathbb C^{N \times N}$.
Note that we can write
\begin{align}
  {\bf X}' = \begin{pmatrix}
       {\bf X}_1 \cr {\bf X}_2 
   \end{pmatrix}  \;,
\end{align}
where ${\bf X}_1 \in \mathbb C^{M \times N}$
and ${\bf X}_2 \in \mathbb C^{(N-M) \times N}$.
From Eq.~\eqref{eq:appmodifiedSylvester} we obtain
\begin{align}
    \begin{pmatrix}
        {\bf A}{\bf X}_1 + {\bf X}_1 {\bf B} \cr c {\bf X}_2  + {\bf X}_2 {\bf B}
    \end{pmatrix} =
    \begin{pmatrix}
       \bC \cr {\bf 0} 
   \end{pmatrix} \;,
\end{align}
and a solution is then
\begin{align}
   {\bf X}' = \begin{pmatrix}
       {\bf X} \cr {\bf 0} 
   \end{pmatrix}  \;.
\end{align}
Furthermore, if we assume that $\gamma_j + \lambda_k \ne 0$
so that the original Sylvester equation had a unique solution, we can choose $c>1$ so that $c + \lambda_k > 0$
as well, implying that Eq.~\eqref{eq:appmodifiedSylvester}
has also a unique solution. (Note that the eigenvalues of ${\bf A}'$ are the $\gamma_j$'s and $c$.) 
In addition, if we let, say, $c > 1 +1/\kappa$, where $\kappa$ is the inverse of the smallest eigenvalue (condition number) of ${\bf Q}:= {\bf A} \otimes \one_N + \one_M \otimes {\bf B}$, then the inverse of the smallest eigenvalue (condition number) of 
${\bf Q}':= {\bf A}' \otimes \one_N + \one_N \otimes {\bf B}$ is only a mild constant factor  
bigger than $\kappa$.

Hence, without loss of generality, it suffices to consider square, diagonalizable matrices of the same dimension for our problem.

%%%%%%%%%%%%%%%%%%%%%%%
\section{Complexity of quantum circuits}
\label{app:complexity}

We provide the proof of our main result in Thm.~\ref{thm:main}. 
To prove the error bounds, we find it easier to prove the corresponding error bounds using vectorization, although some results can be made tighter as we will explain. The following property between norms will be useful.
\begin{lemma}
\label{lem:frobeniuserror}
Let ${\bf Q}\in \mathbb C^{N^2 \times N^2}$ be invertible, 
$\bC\in \mathbb C^{N \times N}$, and 
$\dket{C} \in \mathbb C^{N^2}$ be its vectorized form.
Define $\dket X:= \frac 1 {\bf Q} \dket{C}\in \mathbb C^{N^2}$ and  $\dket {X'}:= h({\bf Q}) \dket{C}\in \mathbb C^{N^2}$, where $h({\bf Q})$ is an approximation to $\frac 1 {\bf Q}$ that satisfies
\begin{align}
\label{eq:lemmainverseapprox}
h(\bQ)=(\one_{N^2}-\varepsilon(\bQ))\frac 1 {\bQ}  , \quad \| \varepsilon(\bQ)\|\le \varepsilon \;,
\end{align}
for some multiplicative error $\varepsilon>0$. Then, if 
${\bf X}\in \mathbb C^{N \times N}$
and ${\bf X'}\in \mathbb C^{N \times N}$ are the corresponding matrix versions of $\dket X$ and $\dket{X'}$, we obtain  
\begin{align}
    \|{\bf X}-{\bf X}'\| \le \varepsilon \sqrt N \|{\bf X}\| \;
\end{align}
for the spectral norms.
\end{lemma}
\begin{proof}
The result follows from a standard property of the Frobenius norm, where for a given ${\bf Y}=\sum_{j,k}y_{jk}\ket j \bra k\in \mathbb C^{N\times N}$,
\begin{align}
    \|{\bf Y}\| \le \|{\bf Y}\|_{\rm Fr}=(\sum_{j,k}|y_{jk}|^2)^{1/2} \le \sqrt N \|{\bf Y}\|\;.
\end{align}
We also note that if $\bf Y$ is vectorized as $\dket Y$, then the Euclidean norm is
\begin{align}
   \| \dket Y \|= \|{\bf Y}\|_{\rm Fr} \; .
\end{align}
In summary, if $h(\bf Q)$ satisfies Eq.~\eqref{eq:lemmainverseapprox}, then
\begin{align}
 \| \bX - \bX' \| \le  \| \bX - \bX' \|_{\rm Fr}   = \| \dket X - \dket {X'} \| 
 = \| \varepsilon(\bQ)  \dket X\|
 \le \varepsilon  \| \dket X\| = \epsilon   \| \bX \|_{\rm Fr} \le \varepsilon \sqrt N \|\bX\| \;.
\end{align}
\end{proof}

We will need to set $\varepsilon = \cO(\epsilon/\sqrt N)$ in the approximations that will guarantee that the normalized solution to Sylvester equation is within additive error $\epsilon$ in spectral norm. This error is exponentially small, but
since our approximations will be exponentially precise, this will incur in an overhead that is logarithmic in $N$. In some cases we might be able to drop the error scaling in $1/\sqrt N$ with refined analyses.

Our results use the LCU approach, for  which in general we can express
the approximated solution as
\begin{align}
    \frac \bX x \approx   \frac {\bX'} x :=\frac 1 x \sum_i x_i V_i (\Pi U_{\bC/\alpha} \Pi) W_i
\end{align}
for unitaries $V_i$ and $W_i$ that correspond to time evolutions with corresponding Hamiltonians determined by $\bA$ and $\bB$; i.e., a linear combination of Hamiltonian simulations. The rescaling factor is $x:=\sum_i |x_i|$. Since we will need to implement these unitaries using a method like QSP, this will incur in an additional error. Note that $\Pi=\frac 1 2 ((2\Pi-\one)+\one)$ is also a simple LCU. 
\begin{lemma}
\label{lem:evolutionerror}
    Let $V'_i$ and $W'_i$ satisfy
    \begin{align}
        \|V_i - V'_i\|\le \epsilon/4 \; , \  \|W_i - W'_i\|\le \epsilon/4 \;,
    \end{align}
    for all $i$, and $1 > \epsilon \ge 0$. Then,
    \begin{align}
        \left \|\frac {\bX'} x -\frac 1 x \sum_i x_i V'_i (\Pi U_{\bC/\alpha} \Pi) W'_i \right \|\le \epsilon \;.
    \end{align}
\end{lemma}
\begin{proof}
    The result is a simple use of the triangle inequality and submultiplicative property of the spectral norm:
    \begin{align}
      \left \|\frac {\bX'} x -\frac 1 x \sum_i x_i V'_i (\Pi U_{\bC/\alpha} \Pi) W'_i \right \| & =  \left \|\frac 1 x \sum_i x_i V_i (\Pi U_{\bC/\alpha} \Pi) W_i -\frac 1 x \sum_i x_i V'_i (\Pi U_{\bC/\alpha} \Pi) W'_i \right \| \\
      & = \frac 1 x \left \|\sum_i x_i (V_i-V'_i)(\Pi U_{\bC/\alpha} \Pi) W_i  + \sum_i x_i V'_i (\Pi U_{\bC/\alpha} \Pi) (W_i-W'_i)\right \| \\
      & \le \frac 1 x \left( \frac \epsilon 4 \sum_i |x_i| + 
      \frac \epsilon 4 \sum_i |x_i| \max_i \|V'\|_i\right) \\
      & \le\frac \epsilon 2 (1+\epsilon) \\
      & \le \epsilon \;.
    \end{align}
\end{proof}
In contrast with the previous result where the precision has to be $\sim 1/\sqrt N$, this Lemma
shows that it suffices to produce the desired unitaries in the LCU within additive error $\cO(\epsilon)$.

Our results involve a discretization of the integral, and approximations using Riemann sums.
We will then make repeated use of the Poisson summation formula and the `Dirac comb', where for given $\Delta>0$ and $x \in \mathbb R$,
\begin{align}
\label{eq:Diraccomb}
    \sum_{j=-\infty}^\infty {\delta}(x-j\Delta) =\frac 1 {\Delta} \sum_{k=-\infty}^\infty
    e^{-\ri 2 \pi k x/\Delta} \;,
\end{align}
Here, $\delta(.)$ is the Dirac delta.

%%%%%%%%%%%%
\subsection{Normal matrices}

In this case ${\bf Q}={\bf Q}_H+i {\bf Q}_S$ satisfies $[{\bf Q},{\bf Q}^\dagger]=[{\bf Q}_H,{\bf Q}_S]=0$.
We also assume $\|\bQ \|\le 1$, and for condition number $\kappa$, we have $\bQ \bQ^\dagger \succeq 1/\kappa^2$. The proof involves several steps and our approach is similar to the one in Ref.~\cite{CKS17}, which uses Eq.~\eqref{eq:Diraccomb}.
The goal is to obtain an approximation to the inverse of the form $h({\bf Q})=(\one_{N^2}-\varepsilon({\bf Q})) \frac 1 {\bf Q}$, where $\|\varepsilon({\bf Q})\|$ is sufficiently small.
We begin with the following result.
\begin{lemma}[Discrete-sum approximation to the inverse of a normal matrix]
\label{lem:LCUnormal}
    Let $\bQ = \bQ_H + \ri \bQ_S \in \mathbb C^{N^2 \times N^2}$ be normal with $\|\bQ\|=1$, condition number $\kappa>0$, and $\varepsilon >0$ the error. Then, there exists positive constants $c_1$, $c_2$, $c_3$, and $c_4$ such that, for  $\delta_t = c_1 \varepsilon/\sqrt{\log(\kappa/\varepsilon)} $, $\delta_\omega = c_2/(\kappa \sqrt{\log(1/\varepsilon)})$, $t_R = c_3 \kappa \sqrt{\log(1/\varepsilon)}$, and $\omega_J = c_4\sqrt{\log(\kappa/\varepsilon)}$, the normal matrix
    \begin{align}
      h(\bQ):= \frac {\ri} {2\pi}\delta_t \delta_\omega^2 \sum_{r=0}^{R-1}   \sum_{j=-J}^J  \sum_{j'=-J}^J (\omega_j - \ri \omega_{j'}) e^{-(\omega_j^2 + \omega_{j'}^2)/2} e^{-\ri t_r (\omega_j \bQ_H + \omega_{j'}\bQ_S)} \;,
    \end{align}
    where $t_r:=r \delta_t$, $\omega_j:=j \delta_\omega$, $R:=\lceil t_R/\delta_t\rceil$, and $J:=\lceil \omega_J/\delta_\omega\rceil$,
    can be written as $h({\bf Q})=(\one_{N^2}-\varepsilon({\bf Q})) \frac 1 {\bf Q}$ and
    \begin{align}
        \| \varepsilon({\bf Q}) \| \le \varepsilon \;.
    \end{align}
\end{lemma}

\begin{proof}
We consider the matrix $f(\bQ)=h(\bQ) \bQ= \one_{N^2} -\varepsilon(\bQ)$ and compare it with $\one_{N^2}$. This can be written as 
$f(\bQ)= f^\rR(\bQ)+ f^\rI(\bQ)$, where in particular
\begin{align}
  f^\rR(\bQ):= \frac {\ri} {2\pi}\delta_t \delta_\omega^2 \sum_{r=0}^{R-1}   \sum_{j=-J}^J  \sum_{j'=-J}^J (\bQ_H \omega_j + \bQ_S \omega_{j'}) e^{-(\omega_j^2 + \omega_{j'}^2)/2} e^{-\ri t_r (\omega_j \bQ_H + \omega_{j'}\bQ_S)} \;.
    \end{align}
    We will show that this term is already close to $\one_{N^2}$.
    We can carry the sum over $r$ exactly
    obtaining
    \begin{align}
 f^\rR(\bQ)= \ri \delta_t  \frac {\delta_\omega^2} {2\pi} \sum_{j=-J}^J 
     \sum_{j'=-J}^J   ( \bQ_H \omega_j+ \bQ_S \omega_{j'}) e^{-(\omega_j^2+\omega_{j'}^2)/2}  \frac{\one_N - 
     e^{-\ri t_R (\omega_j \bQ_H +  \omega_{j'} \bQ_S)}}{\one_N - 
     e^{-\ri \delta_t (\omega_j \bQ_H +  \omega_{j'} \bQ_S)}} \;.
\end{align}
In this expression, the limit must be taken if
$(\omega_j \bQ_H +  \omega_{j'} \bQ_S) \rightarrow 0$.
Under the assumptions, the parameters obey $\|\delta_t (\omega_j \bQ_H +  \omega_{j'} \bQ_S)\| \le 2\delta_t \omega_J =\cO( \varepsilon)$, and we can perform a Taylor series expansion obtaining
\begin{align}
    \frac 1 {\one_N - 
     e^{-\ri \delta_t (\omega_j \bQ_H +  \omega_{j'} \bQ_S)}} = \frac 1 {\ri \delta_t (\omega_j \bQ_H +  \omega_{j'} \bQ_S) } \left( \one_N-\ri \delta_t (\omega_j \bQ_H +  \omega_{j'} \bQ_S)/2 + \ldots \right) \;.
\end{align}
Hence,
\begin{align}
     f^\rR(\bQ)=    \frac {\delta_\omega^2} {2\pi} \sum_{j=-J}^J 
     \sum_{j'=-J}^J     e^{-(\omega_j^2+\omega_{j'}^2)/2}  \left( \one_N+ \cO(\varepsilon) \right)  \left(\one_N - 
     e^{-\ri t_R (\omega_j \bQ_H +  \omega_{j'} \bQ_S)}\right) \;.
\end{align}
      Then, up to error $\cO(\varepsilon)$, our approximation $f^\rR(\bQ)$ is
     simplified to
     \begin{align}
     \label{eq:frapprox}
           \frac {\delta_\omega^2} {2\pi} \sum_{j=-J}^J 
     \sum_{j'=-J}^J    e^{-(\omega_j^2+\omega_{j'}^2)/2}  (\one_{N^2} - 
     e^{-\ri t_R (\omega_j \bQ_H +  \omega_{j'} \bQ_S)})  \;.
     \end{align}

     Our next goal is to show that this is $\cO(\varepsilon)$-close to $\one_{N^2}$.
To this end we use Eq.~\eqref{eq:Diraccomb}
and the Fourier transform for the Gaussian and show for the infinite sum
\begin{align}
\frac{   \delta_\omega}{\sqrt{2\pi}} \sum_{j=-\infty}^\infty    e^{-\omega_j^2 /2} & = \frac{   \delta_\omega}{\sqrt{2\pi}}\int_{-\infty}^\infty \rd \omega \; e^{-\omega^2 /2} \sum_{j=-\infty}^\infty \delta(\omega-\omega_j)\\
&=\sum_{k=-\infty}^\infty e^{-(2 \pi k/\delta_\omega)^2/2} \\
\label{eq:discreteGaussiansum}
& =1 + 2 \sum_{k =1}^\infty  e^{-(2 \pi k/\delta_\omega)^2/2}\;.
\end{align}
The correction to 1 is exponentially small in $1/\delta_\omega^2$. Hence, it suffices to choose some $\delta_\omega = c_2/\sqrt{\log(1/\varepsilon)}$ for some constant $c_2>0$ to obtain error $\cO(\varepsilon)$. That is,
\begin{align}
      \frac {\delta_\omega^2} {2\pi} \sum_{j=-\infty}^\infty 
     \sum_{j'=-\infty}^\infty     e^{-(\omega_j^2+\omega_{j'}^2)/2}  -\one_{N^2}=\cO(\varepsilon) \;.
\end{align}
We also consider the error by cutting off the infinite sum, which can be bounded since
\begin{align}
   \frac{   \delta_\omega}{\sqrt{2\pi}} \sum_{j:|j|>J}   e^{-\omega_j^2 /2} =\cO(  e^{-\omega_J^2 /2}) =\cO(\varepsilon/\kappa)=\cO(\varepsilon) \;,
\end{align}
under the hypothesis. These imply for the first term of Eq.~\eqref{eq:frapprox},
    \begin{align}
           \frac {\delta_\omega^2} {2\pi} \sum_{j=-J}^J 
     \sum_{j'=-J}^J    e^{-(\omega_j^2+\omega_{j'}^2)/2}  -\one_{N^2}=\cO(\varepsilon)  \;,
     \end{align}
as desired.

We now consider the second term of Eq.~\eqref{eq:frapprox}. Note that, for example, using Eq.~\eqref{eq:Diraccomb} the infinite sum obeys
\begin{align}
\frac{   \delta_\omega}{\sqrt{2\pi}} \sum_{j=-\infty}^\infty    e^{-\omega_j^2 /2} e^{-\ri t_R \omega_j \bQ_H}& = \frac{   \delta_\omega}{\sqrt{2\pi}}\int_{-\infty}^\infty \rd \omega \; e^{-\omega^2 /2} e^{-\ri t_R \omega \bQ_H} \sum_{j=-\infty}^\infty \delta(\omega-\omega_j)\\
&=\sum_{k=-\infty}^\infty e^{-(t_R \bQ_H +2 \pi k/\delta_\omega)^2/2} \\
& = e^{-(t_R \bQ_H)^2/2}+ \sum_{k\ne 0} e^{-(t_R \bQ_H +2 \pi k/\delta_\omega)^2/2} \;.
\end{align}
The second term is sufficiently small with the right choice of constants where, for example, $2\pi/\delta_\omega \ge t_R$. With this assumption the error is exponentially small in $1/\delta_\omega^2$, and hence $\cO(\varepsilon)$.
In addition, like for the previous term, we can place a cutoff in the sum at $|j|\le J$ introducing error at most $\cO(\varepsilon)$. These bounds imply for the second term of Eq.~\eqref{eq:frapprox},
\begin{align}
  \|    \frac {\delta_\omega^2} {2\pi} \sum_{j=-J}^J 
     \sum_{j'=-J}^J    e^{-(\omega_j^2+\omega_{j'}^2)/2}    
     e^{-\ri t_R (\omega_j \bQ_H +  \omega_{j'} \bQ_S)} \| =
     e^{-(t_R \bQ_H)^2/2} e^{-(t_R \bQ_S)^2/2}+\cO(\varepsilon) \;.
\end{align}
Furthermore, our assumptions imply $\| e^{-(t_R \bQ_H)^2/2} e^{-(t_R \bQ_S)^2/2}\|=\cO(\varepsilon)$, as desired.

The other term in  $f(\bQ)$ is
\begin{align}
    f^\rI(\bQ):= \frac {\ri} {2\pi}\delta_t \delta_\omega^2 \sum_{r=0}^{R-1}   \sum_{j=-J}^J  \sum_{j'=-J}^J (\ri \bQ_S \omega_j -\ri \bQ_H \omega_{j'}) e^{-(\omega_j^2 + \omega_{j'}^2)/2} e^{-\ri t_r (\omega_j \bQ_H + \omega_{j'}\bQ_S)} \;.
\end{align}
We will show this is $\cO(\varepsilon)$.
First we approximate by performing an infinite sum; for example, using Eq.~\eqref{eq:Diraccomb},
\begin{align}
  \frac{\delta_\omega}{\sqrt{2\pi}}  \sum_{j=-\infty}^\infty e^{-\omega_j^2/2} 
    e^{-\ri t_r \omega_j \bQ_H} & = \sum_{k=-\infty}^\infty e^{-(t_r \bQ_H + 2 \pi k/\delta_\omega)^2/2} \;,\\
    \frac{\delta_\omega}{\sqrt{2\pi}}  \sum_{j=-\infty}^\infty \omega_j e^{-\omega_j^2/2} 
    e^{-\ri t_r \omega_j \bQ_H} & = -\ri \sum_{k=-\infty}^\infty (t_r \bQ_H ) e^{-(t_r \bQ_H + 2 \pi k/\delta_\omega)^2/2} \;.
\end{align}
It follows that
\begin{align}
  \frac {-\ri} {2\pi}  \delta_\omega^2   
  \sum_{j=-\infty}^\infty  \sum_{j'=-\infty}^\infty (\ri \bQ_S \omega_j -\ri \bQ_H \omega_{j'}) e^{-(\omega_j^2 + \omega_{j'}^2)/2} e^{-\ri t_r (\omega_j \bQ_H + \omega_{j'}\bQ_S)}=0 \;. 
\end{align}
Hence, we can replace the sum over $j,j'$ in $f^\rI(\bQ)$ so that
\begin{align}
\nonumber
 f^\rI(\bQ)&= \frac {\ri} {2\pi}\delta_t \delta_\omega^2 \sum_{r=0}^{R-1}  \left( \sum_{j:|j|>J}  \sum_{j'=-J}^J (\ri \bQ_S \omega_j -\ri \bQ_H \omega_{j'}) e^{-(\omega_j^2 + \omega_{j'}^2)/2} e^{-\ri t_r (\omega_j \bQ_H + \omega_{j'}\bQ_S)} \right . \\ 
 &+ \left.
 \sum_{j=-\infty}^\infty  \sum_{j':|j'|>J} (\ri \bQ_S \omega_j -\ri \bQ_H \omega_{j'}) e^{-(\omega_j^2 + \omega_{j'}^2)/2} e^{-\ri t_r (\omega_j \bQ_H + \omega_{j'}\bQ_S)}\right) \;.
\end{align}
We can upper bound the spectral norm of each term individually and use the multiplicative property, obtaining
\begin{align}
\nonumber
\| f^\rI(\bQ)\|&\le \frac {1} {\pi}\delta_t R\delta_\omega^2   \left( \sum_{j:|j|>J}  \sum_{j'=-J}^J ( |\omega_j| + |\omega_{j'}|) e^{-(\omega_j^2 + \omega_{j'}^2)/2} +
 \sum_{j=-\infty}^\infty  \sum_{j':|j'|>J} ( |\omega_j| + |\omega_{j'}|) e^{-(\omega_j^2 + \omega_{j'}^2)/2}\right) \;.
 \end{align}
 Since $\delta_\omega \sum_j |\omega_j| e^{-\omega_j^2/2}=\cO(1)$ and $\delta_\omega \sum_{|j|>J} |\omega_j| e^{-\omega_j^2/2}=\cO(e^{-\omega_J^2/2})$ then
 \begin{align}
   \| f^\rI(\bQ)\| = \cO(t_R e^{-\omega_J^2/2}) \;. 
 \end{align}
 We can choose the constants such that this error is also $\cO(\varepsilon)$.
 Then, $\|f^\rI(\bQ)\|=\cO(\varepsilon)$.

It follows that $\|\varepsilon(\bQ)\| = \|\one_N -f(\bQ)\| \le \varepsilon$ with a proper choice of constants.

\end{proof}

\vspace{.5cm}

The approximation $h(\bQ)$ is readily an LCU that involves evolutions under $\bQ_H$ and $\bQ_S$.
Next we consider the $L_1$-norm for this approximation.
\begin{lemma}
    The $L_1$-norm of  $h(\bQ)$ in Lemma~\ref{lem:LCUnormal} satisfies
    \begin{align}
     y:=   \frac 1 {2\pi} \delta_t \delta_{\omega}^2 \sum_{r=0}^{R-1} \sum_{j=-J}^J \sum_{j'=-J}^J
        (|\omega_j|+|\omega_{j'}|)e^{-(\omega_j^2+\omega_{j'}^2)/2}
        = \cO(\kappa \sqrt{\log(1/\varepsilon)}) \;.
    \end{align}
\end{lemma}
\begin{proof}
    This $L_1$ norm is upper bounded by that where the sums over $j$ and $j'$ are infinite.
    For example, we can consider Eq.~\eqref{eq:discreteGaussiansum}
    that shows 
    $\frac 1 {\sqrt{2\pi}} \delta_\omega\sum_{j=-\infty}^\infty e^{-\omega_j^2/2}$
    is $\cO(\varepsilon)$-close to 1 for our choice of parameters.
    Then
    \begin{align}
        \frac 1 {\pi}\delta_{\omega}^2 \sum_{j=-\infty}^\infty \sum_{j'=-\infty}^\infty
        |\omega_j| e^{-(\omega_j^2+\omega_{j'}^2)/2} \approx
        2 \frac 1 {\sqrt{2\pi}} \delta_\omega
        \sum_{j=-\infty}^\infty 
         |\omega_j| e^{-\omega_j^2/2} \approx 
          4 \frac 1 {\sqrt{2\pi}} \int_0^\infty \rd \omega\;
          \omega  e^{-\omega^2/2}=\frac{4}{\sqrt{2\pi}}\;,
    \end{align}
    where the first approximation error is $\cO(\varepsilon)$ and the second approximation error is $\cO(\delta_\omega)$.
    Hence, the $L_1$ norm is bounded by $c R \delta_t = ct_R$, for some constant $c>0$, and this is $\cO(\kappa \sqrt{\log(1/\varepsilon)})$.
\end{proof}

We are now ready to determine the complexity of our algorithm and prove the first case in Thm.~\ref{thm:main}. Applying the inverse of vectorization, our approximate solution is
\begin{align}
   \bX \approx {\bf X}' := \frac{\ri}{2\pi} \delta_t \delta_\omega^2 \sum_{r=0}^{R-1} \sum_{j=-J}^J \sum_{j'=-J}^J (\omega_j -\ri \omega_{j'}) e^{-(\omega_j^2 + \omega_{j'}^2)^2/2} e^{-\ri t_r (\omega_j \bA_H + \omega_{j'}\bA_S)} \bC e^{-\ri t_r (\omega_j \bB_H + \omega_{j'}\bB_S)} \;.
\end{align}
If we replace $\bC$ by its block-encoding $ U_{\bC/\alpha}$, 
then this expression becomes an LCU. 
Our quantum algorithm then prepares the block-encoding of $\bX'/x$, where $x:=y \alpha =\cO(\alpha \kappa \sqrt{\log(1/\varepsilon)})$, 
\begin{align}
   U_{\bX/x}  \implies \Pi  U_{\bX/x} \Pi = \bX'/x \;.
\end{align}
This can be done via  Lemma~\ref{lem:BEfromLCU}.
Note that Lemma~\ref{lem:frobeniuserror} gives
\begin{align}
    \|\bX - \bX'\|\le \varepsilon \sqrt N \|\bX\| \implies \|\bX/x - \Pi U_{\bX/x}\Pi\| \le 
    \varepsilon \sqrt N/x \;.
\end{align}
Since the multiplicative error is $\varepsilon \sqrt N$, we can set $\varepsilon=\cO(\epsilon/\sqrt N)$ for multiplicative error $\cO(\epsilon)$. When rescaling by $x \ge \|\bX\|$, the multiplicative error becomes an additive error also $\cO(\epsilon)$. Then, $x =\cO(\alpha \kappa \sqrt{\log(N/\epsilon)})$.

The query complexity can be determined as follows. 
Since each term in the LCU uses $U_{\bC/\alpha}$ once, we can implement $ U_{\bX/x}$ with a single use of $U_{\bC/\alpha}$; see Fig.~\ref{fig:normalcase} for an example. The query complexity will then be dominated by the largest evolution time in the LCU, which for this case is linear in
$t_R \omega_J = \cO(\kappa \log(\kappa/\varepsilon)$. Using QSP to simulate each time evolution within precision $\cO(\epsilon)$ from Lemma~\ref{lem:evolutionerror}, the overall query complexity is then
\begin{align}
    Q= \cO \left(\kappa \log (\kappa/\varepsilon) + \log(1/\epsilon)\right) 
   =\cO \left(\kappa \log(  \kappa N/ \epsilon) \right) \;.
\end{align}

The additional gate complexity is to implement $V$,
which is the unitary acting on an ancillary register
that sets the weights in the LCU. Note that we can express
\begin{align}
    U_{\bX/x}=\frac 1 x \sum_{i} x_i V_i U_{\bC/\alpha} W_i 
\end{align}
where the $V_i$'s and $W_i$'s are the time evolutions. We can assume $x_i \ge 0$ by absorbing the phase in the definition of the unitaries.
Then,
\begin{align}
    V \ket 0 \mapsto \frac 1 {\sqrt x} \sum_i \sqrt {x_i}\ket i \;.
\end{align}
 For this case the register $\ket i$ contains three registers, one to encode the times $t_r$, another to encode the frequencies $\omega_j$, and another for the frequencies $\omega_{j'}$. The time register can be prepared in a uniform superposition with $\log_2 R$ gates by simply applying Hadamard gates. The frequency registers are in states with Gaussian-like amplitudes or amplitudes that correspond to the first Hermite functions. We can use the results in Ref.~\cite{Som15} to prepare these states with complexity $\cO(\log(J/\epsilon))$, since it suffices to prepare these states within error $\epsilon$.
Hence, the gate complexity is
\begin{align}
    G=\cO( Q+\log(RJ/\epsilon))=\cO(Q+\log(\kappa/\varepsilon))
    =
    \cO(\kappa \log (   \kappa N/ \epsilon))\;.
\end{align}
The term $Q$ is due to the gates in QSP.

%%%%%%%%%%%%%%%%%%%%%%%%%%%%%%%%%%
\subsection{Matrices with positive Hermitian part}

In this case $\bQ=\bQ_H + \ri \bQ_S$ and we assume $\bQ_H \succ 0$ to be strictly positive. In fact,
we are going to assume $\bQ_H \succeq (1/\kappa)\one_{N^2}$, for some $\kappa>0$ that 
is related to the condition number since 
$\|\bQ\|\le \|\bA\|+\|\bB\|\le 1$. 
Equation~\eqref{eq:inversewithexponential} gives
\begin{align}
    \dket{X}=\int_0^\infty \rd t \; e^{-t \bQ} \dket C \implies \bX = \int_0^\infty \rd t \; e^{-t{\bA}}\bC e^{-t{\bB}} \;.
\end{align}
Hence, we form an approximation $\tilde{\bf X}$ of ${\bf X}$ by truncating the integral and discretizing it by a Riemann sum. The truncation we require is given by the following result.
\begin{lemma}[Approximation to $\bX$ for the case of a matrix with positive Hermitian part]
\label{lem:LCUpositiveHermitian}
    Let $\bQ = \bQ_H + \ri \bQ_S \in \mathbb C^{N^2 \times N^2}$ have a positive Hermitian component $\bQ_H\succeq\frac{1}{\kappa}\one_{N^2}$, where $\kappa>0$, and $\epsilon >0$ the error. Then, there exists positive constants $c_1$, $c_2$, $c_3$, $c_4$, and $\beta\in(0,1)$ such that, for  $\delta_t = c_1 \epsilon/(\kappa\|{\bf C}\|) $, $\delta_\omega = c_2\epsilon/(\omega_J\|{\bf C}\|t_R^2)$, $t_R = c_3 \kappa \log\frac{\kappa\|{\bf C}\|}{\epsilon}$, and $\omega_J =c_4\log^{1/\beta}\frac{t_R\|{\bf C}\|}{\epsilon}$, the matrix
    \begin{align}
      \tilde \bX:= 
     \delta_t\delta_\omega\sum_{r=0}^{R-1}\sum_{j=-J}^J\hat{f}(\omega_j)e^{-\ri ({\bf A}_H\omega_j+{\bf A}_S)t_r}{\bf C}e^{-\ri ({\bf B}_H\omega_j+{\bf B}_S)t_r},
    \end{align}
    where $t_r:=r \delta_t$, $\omega_j:=j \delta_\omega$, $R:=\lceil t_R/\delta_t\rceil$, and $J:=\lceil \omega_J/\delta_\omega\rceil$,
     and if $\epsilon({\bf{X}}):=\tilde \bX - \bX$,
    \begin{align}
        \|\epsilon({\bf{X}}) \| \le \epsilon \;.
    \end{align}
\end{lemma}
\begin{proof}
We first form an approximation ${\bf X}'$ by truncating the integral to finite $T< \infty$: 
\begin{align} 
\|{\bf X}'-{\bf X}\|&=\left\|\int_0^{T} \rd t \; e^{-t{\bf A}}\bC e^{-t{\bf B}} -\int_0^\infty \rd t \; e^{-t{\bf A}}\bC e^{-t{\bf B}} \right\|
\\
&=\left\|\int_{T}^\infty \rd t \; e^{-t{\bf A}}\bC e^{-t{\bf B}} \right\|  \\
&\le\|\bC\|\int_{T}^\infty \rd t \; \| e^{-t{\bf A}}\|\|e^{-t{\bf B}}\| \\
& =\|\bC\|\int_{T}^\infty \rd t \; \| e^{-t{\bf Q}}\|\;,
\end{align}
where the last equality follows from
\begin{align}
\|e^{-t{\bf Q}}\|=\|e^{-t(({\bf A}_H+i{\bf A}_S)\otimes\one_N+\one_N\otimes({\bf B}_H^\rT+i{\bf B}_S^\rT))}\|=\|e^{-t({\bf A}_H+i{\bf A}_S)}\|\|e^{-t({\bf B}_H^\rT+i{\bf B}_S^\rT)}\|
=\|e^{-t{\bf A}}\|\|e^{-t{\bf B}^\rT }\|
=\|e^{-t{\bf A}}\|\|e^{-t{\bf B}}\|\;.
%\le\|e^{-{\bf A}_H s}\|\|e^{-{\bf B}_H^\rT s}\|.
\end{align}
By assumption, ${\bf Q}\succeq\frac{1}{\kappa}\one_{N^2}$. Hence $\|e^{-t{\bf{Q}}}\|\le e^{-t/\kappa}$ and
\begin{align}
\|{\bf X}'-{\bf X}\|
\le\|\bC\|\int_{T}^\infty \rd t \; e^{-t/\kappa}
 =\|\bC\|{\kappa}e^{-T/\kappa}:=\epsilon_0.
\end{align}
Hence, an $\epsilon_0$-approximate solution is obtained by choosing $t_R=t=\mathcal{O}(\kappa\log\frac{\kappa\|\bC\|}{\epsilon_0})$.
Let us now discretize the integral by a Riemann sum at $R$ points $t_r=r\delta_t$. This introduces another error
\begin{align} 
\|{\bf X}''-{\bf X}'\|&=\left\|\int_0^{T} \rd t \; e^{-t{\bf A}}\bC e^{-t{\bf B}} -\delta_t\sum_{r=0}^{R-1}e^{-t_r{\bf A} }\bC e^{-t_r{\bf B}}\right\|
\\ 
&\le\sum_{r=0}^{R-1}\int_0^{\delta_t} \rd t \; \left\|e^{-{(t_r+t) \bf A}}\bC e^{-{(t_r+t) \bf B}}-e^{-t_r {\bf A}}\bC e^{-t_r{\bf B} }\right\| 
\\ 
&\le\|{\bf C}\|\sum_{r=0}^{R-1}\int_0^{\delta_t} \rd t \; \|e^{-{\bf A}(t_r+t)}\| \|e^{-{\bf B}(t_r+t)}- e^{-{\bf B}t_r}\|+
 \|e^{-{\bf A}(t_r+t)}-e^{-{\bf A}t_r}\|\| e^{-{\bf B}t_r}\|
 \\  
&=\mathcal{O}\left(\|{\bf C}\|\sum_{r=0}^{R-1}\|e^{-{\bf A}t_r}\|\|e^{-{\bf B}t_r}\|\int_0^{\delta_t}(\|{\bf B}\|+\|{\bf A}\|
)s\right)
 \mathrm{d}s
 \\ 
&=\mathcal{O}\left(\delta_t^2\|{\bf C}\|\sum_{r=0}^{R-1}\|e^{-{\bf A}t_r}\|\|e^{-{\bf B}t_r}\|(\|{\bf B}\|+\|{\bf A}\|
)\right) 
\\
&=\mathcal{O}\left(\delta_t^2\|{\bf C}\|(\|{\bf B}\|+\|{\bf A}\|)\sum_{r=0}^{R-1}e^{-\frac{1}{\kappa}t_r}\right)
\\ 
&
=\mathcal{O}\left(\delta_t\kappa\|{\bf C}\|(\|{\bf B}\|+\|{\bf A}\|)\right)
\\
&=\mathcal{O}\left(\delta_t\kappa\|{\bf C}\|\right)
=\mathcal{O}(\epsilon_1).
\end{align}
Hence, discretization error in $t$ is bounded by $\epsilon_1$ with the choice $\delta_t=\mathcal{O}(\frac{\epsilon_1}{\kappa\|{\bf C}\|})$

Let $\hat{f}(\omega)$ be a kernel function that realizes exponential decay
\begin{align}
\forall t\ge0, \quad e^{-t}=\int_\mathbb{R}\hat{f}(\omega)e^{-\ri\omega t}\mathrm{d}\omega.
\end{align}
By the Linear-Combination-of-Hamiltonian-Simulation (LCHS) technique~\cite{an2023linear,an2023quantum}, there exist choices of $\hat{f}(\omega)$ that are bounded like $|\hat{f}(\omega)|=\mathcal{O}(1)$ and smooth like $|\hat{f}'(\omega)|=\mathcal{O}(1)$ such that for any matrix ${\bf Q}$ satisfying ${\bf Q}_H\succeq 0$, and any $t\ge 0$,
\begin{align}\label{eq:lchs}
e^{-{\bf Q}t}=e^{-({\bf Q}_H+\ri{\bf Q}_S)t}=\int_\mathbb{R}\hat{f}(\omega)e^{-\ri({\bf Q}_H\omega+{\bf Q}_S)t}\mathrm{d}\omega.
\end{align}
By `unvectorizing' the expression, we obtain the equalities
\begin{align}
\mathrm{unvec}[e^{-{\bf Q}t}\dket{C}]&=\int_{\mathbb{R}} \rd \omega \; \hat{f}(\omega)e^{-\ri({\bf A}_H\omega+{\bf A}_S)t}{\bf C}e^{-\ri({\bf B}_H\omega+{\bf B}_S)t} ,
\\
{\bf X}=\mathrm{unvec}\left[\int_0^\infty e^{-{\bf Q}s}\mathrm{d}s\dket{C}\right]&=\int_{0}^\infty\int_{\mathbb{R}}\hat{f}(\omega)e^{-\ri ({\bf A}_H\omega+{\bf A}_S)s}{\bf C}e^{-\ri ({\bf B}_H\omega+{\bf B}_S)s}\mathrm{d}\omega\mathrm{d}s,
\\
{\bf X}''&=\delta_t\sum_{r=0}^{R-1}\int_{\mathbb{R}}\hat{f}(\omega)e^{-\ri ({\bf A}_H\omega+{\bf A}_S)t_r}{\bf C}e^{-\ri ({\bf B}_H\omega+{\bf B}_S)t_r}\mathrm{d}\omega
\end{align}
where ${\rm unvec}$ refers to the inverse of the vectorization.
According to Lemma 9 of~\cite{an2023quantum}, there exists a choice of $\hat{f}_\beta(\omega)$ parameterized by any constant $\beta\in(0, 1)$, that satisfies~\cref{eq:lchs} and decays such that the tails of the integral truncated to outside the interval $[-\omega_J,\omega_J]$ where $\omega_J=\mathcal{O}(\log^{1/\beta}\frac{1}{\epsilon_2'})$, satisfies $\int_{\mathbb{R}\backslash[-\omega_J,\omega_J]}|\hat{f}_\beta(\omega)|\mathrm{d}\omega\le\epsilon_2'$ for any $\epsilon_2'>0$.
Let ${\bf X}'''$ be the truncation of the inner integral. Then
\begin{align}\nonumber
\|{\bf X}'''-{\bf X}''\|&=\delta_t\sum_{r=0}^{R-1}\left\|\int_{\mathbb{R}\backslash[-\omega_J,\omega_J]}\hat{f}(\omega)e^{-\ri ({\bf A}_H\omega+{\bf A}_S)t_r}{\bf C}e^{-\ri ({\bf B}_H\omega+{\bf B}_S)t_r}\mathrm{d}\omega\right\|
\\\nonumber
&\le\delta_t\sum_{r=0}^{R-1}\int_{\mathbb{R}\backslash[-\omega_J,\omega_J]}|\hat{f}(\omega)|\left\|e^{-\ri ({\bf A}_H\omega+{\bf A}_S)t_r}{\bf C}e^{-\ri ({\bf B}_H\omega+{\bf B}_S)t_r}\right\|\mathrm{d}\omega
\\
&\le\delta_t\|{\bf C}\|\sum_{r=0}^{R-1}\int_{\mathbb{R}\backslash[-\omega_J,\omega_J]}|\hat{f}(\omega)|\mathrm{d}\omega\le t_R\|{\bf C}\|\epsilon_2'=\mathcal{O}(\epsilon_2).
\end{align}
Hence, it suffices to choose $\omega_J=\mathcal{O}(\log^{1/\beta}\frac{t_R\|{\bf C}\|}{\epsilon_2})$.

Now, let ${\bf X}''''$ be the discretization of the inner integral to finite number of $2J$ points $\omega_j=\delta_\omega j$ for $j=-J,...,J-1$, with step size $\delta_\omega=\mathcal{O}(\omega_J/J)$.
Then by a Taylor expansion,
\begin{align}\nonumber
&\|{\bf X}''''-{\bf X}'''\|\\\nonumber
&=\delta_t\sum_{r=0}^{R-1}\left\|\sum_{j=-J}^{J-1}\int_{0}^{\delta_\omega}\hat{f}(\omega_j+s)e^{-\ri ({\bf A}_H(\omega_j+s)+{\bf A}_S)t_r}{\bf C}e^{-\ri ({\bf B}_H(\omega_j+s)+{\bf B}_S)t_r}-\hat{f}(\omega_j)e^{-\ri ({\bf A}_H\omega_j+{\bf A}_S)t_r}{\bf C}e^{-\ri ({\bf B}_H\omega_j+{\bf B}_S)t_r}\mathrm{d}s\right\|
\\\nonumber
&\le\delta_t\sum_{r=0}^{R-1}\sum_{j=-J}^{J-1}\int_{0}^{\delta_\omega}\left\|\hat{f}(\omega_j+s)e^{-\ri ({\bf A}_H(\omega_j+s)+{\bf A}_S)t_r}{\bf C}e^{-\ri ({\bf B}_H(\omega_j+s)+{\bf B}_S)t_r}-\hat{f}(\omega_j)e^{-\ri ({\bf A}_H\omega_j+{\bf A}_S)t_r}{\bf C}e^{-\ri ({\bf B}_H\omega_j+{\bf B}_S)t_r}\right\|\mathrm{d}s
\\\nonumber
&=\mathcal{O}\left(\delta_t\|{\bf C}\|\sum_{r=0}^{R-1}\sum_{j=-J}^{J-1}\int_{0}^{\delta_\omega}s(|\hat{f}'(\omega_j)|+|\hat{f}(\omega_j)|(\|{\bf A}_H\|+\|{\bf B}_H\|)t_r)\mathrm{d}s\right)
\\\nonumber
&=\mathcal{O}\left(\delta_t\|{\bf C}\|\sum_{r=0}^{R-1}\sum_{j=-J}^{J-1}\int_{0}^{\delta_\omega}s(1+(\|{\bf A}_H\|+\|{\bf B}_H\|)t_r)\mathrm{d}s\right)
\\\nonumber
&=\mathcal{O}\left(\omega_J\|{\bf C}\|\delta_\omega(\|{\bf A}_H\|+\|{\bf B}_H\|)t_R^2\right)=\mathcal{O}(\epsilon_3)
\\
&=\mathcal{O}\left(\omega_J\|{\bf C}\|\delta_\omega t_R^2\right)=\mathcal{O}(\epsilon_3),
\end{align}
it suffices to choose $\delta_\omega=\epsilon_3/(\omega_J\|{\bf C}\|t_R^2)$. Note that we may extend the range of the sum from $J-1$ to $J$ without changing the asymptotic error scaling.
We complete our proof by a triangle inequality to show that $\|{\bf X}''''-{\bf X}\|=\mathcal{O}(\epsilon_0+\epsilon_1+\epsilon_2+\epsilon_3)$ and choosing all $\epsilon_0=\epsilon_1=\epsilon_2=\epsilon_3=\epsilon$.
\end{proof}
We note that it is possible to improve the scaling of $\delta_t$ and $\delta_\omega$ with $\epsilon$ to $\log(\frac{1}{\epsilon})$, by a better analysis or by optimal non-uniform Gaussian quadrature with non-uniform weights~\cite{an2023quantum}, but this does not affect the overall query complexity, and only improves gate complexity by a logarithmic factor.

The approximation $\tilde \bX$ is already an LCU, i.e., a linear combination of Hamiltonian evolutions under $\bA$ and $\bB$.
Next we consider the L1-norm of this LCU.
\begin{lemma}
    The $L_1$-norm of $\tilde \bX$ in Lemma~\ref{lem:LCUpositiveHermitian} satisfies
    \begin{align}
        y:= \delta_t\delta_\omega\sum_{r=0}^{R-1}\sum_{j=-J}^J|\hat{f}(\omega_j)|=\mathcal{O}\left(\kappa\log\frac{\kappa\|{\bf C}\|}{\epsilon}\right)
    \end{align}
\end{lemma}
\begin{proof}
The sum over $r$ is easy to evaluate: 
\begin{align}
y:= \delta_tR\delta_\omega\sum_{j=-J}^J|\hat{f}(\omega_j)|=t_R\delta_\omega\sum_{j=-J}^J|\hat{f}(\omega_j)|=\mathcal{O}\left(\kappa\log\frac{\kappa\|{\bf C}\|}{\epsilon}\delta_\omega\sum_{j=-J}^J|\hat{f}(\omega_j)|\right).
\end{align}
Hence, it suffices to show that $\delta_\omega\sum_{j=-J}^J|\hat{f}(\omega_j)|=\mathcal{O}(1)$.
The specific choice of $|\hat{f}(\omega)|$ in the proof of~\cref{lem:LCUpositiveHermitian} is
\begin{align}
\hat{f}(\omega)=\frac{1}{2\pi}\frac{1}{1-i\omega}\frac{1}{e^{(1+i\omega)^\beta}},
\quad|\hat{f}(\omega)|=\frac{e^{-\left(\omega^2+1\right)^{\beta /2} \cos (\beta  \tan^{-1}\omega)}}{2 \pi  \sqrt{\omega^2+1}}.
\end{align}
According to~\cite{an2023quantum}, for any constant $\beta\in(0,1)$, the integral
\begin{align}
\int_\mathbb{R}|\hat{f}(\omega)|\mathrm{d}\omega=\mathcal{O}(1).
\end{align}
This can be understood as arising from the superpolynomial decay of the numerator.
Moreover, as $\hat{f}(\omega)$ is a smooth function, $|\hat{f}'(\omega)|$ is bonded by a constant. Hence the error between the integral and its discrete approximation is
\begin{align}\nonumber
&\int_{-\omega_J}^{\omega_J}|\hat{f}(\omega)|\mathrm{d}\omega-\delta_\omega\sum_{j=J}^{J-1}|\hat{f}(\omega_j)|
=\sum_{j=-J}^{J-1}\int_0^{\delta_\omega}|\hat{f}(\omega_j+s)|-|\hat{f}(\omega_j)|\mathrm{d}s
=\mathcal{O}\left(\sum_{j=-J}^{J-1}\int_0^{\delta_\omega}|s\hat{f}'(\omega_j)|\mathrm{d}s\right)
\\
&=\mathcal{O}\left(J\delta_\omega^2\right)
=\mathcal{O}\left(\omega_J\delta_\omega\right)
=\mathcal{O}\left(\epsilon /(\|{\bf C}\|t_R^2)\right)=\mathcal{O}(\epsilon),
\end{align}
and is arbitrarily small, where we substitute the choice of $\delta_\omega$ in~\cref{lem:LCUpositiveHermitian}. We complete the proof by a triangle inequality.
\end{proof}

We are ready to obtain the query and gate complexities for this case.  The query complexity is then determined by the largest evolution time in the LCU. This is linear in $t_R \omega_J =\cO(\kappa \log^{1+1/\beta}(\frac{\|{\bf C}\|\kappa}{\epsilon}))$. We can use a method like QSP to simulate this evolution, and it suffices to do so within additive error $\cO(\epsilon)$ due to Lemma~\ref{lem:evolutionerror}. This gives the query complexity
\begin{align}
    Q= \cO\left(\kappa \log^{1+1/\beta}\frac{\|{\bf C}\|\kappa}{\epsilon}+ \log(1/\epsilon)\right)=
    \cO\left(\kappa \log^{1+1/\beta}\frac{\|{\bf C}\|\kappa}{\epsilon}\right) \;.
\end{align}

The additional gate complexity is to implement $V$,
which is the unitary acting on an ancillary register
that sets the weights in the LCU. Note that we can express
\begin{align}
    U_{\bX/x}=\frac 1 x \sum_{i} x_i V_i U_{\bC/\alpha} W_i 
\end{align}
where the $V_i$'s and $W_i$'s are the time evolutions. We can assume $x_i \ge 0$ by absorbing the phase in the definition of the unitaries.
Then,
\begin{align}
    V \ket 0 \mapsto \frac 1 {\sqrt x} \sum_i \sqrt {x_i}\ket i \;.
\end{align}
 For this case the register $\ket i$ contains two registers, one to encode the times $t_r$, another to encode the frequencies $\omega_j$. The time register can be prepared in a uniform superposition with $\log_2 R$ gates by simply applying Hadamard gates. The frequency registers are in states with amplitudes corresponding to the kernel function $\hat{f}(\omega)$. As $\hat{f}(\omega)$ is analytic for real $\omega$, its function values can be efficiently computed. In particular, the superposition state of amplitudes can be prepared~\cite{pocrnic2025constantfactorimprovementsquantumalgorithms} to within error $\epsilon$ using $\cO(\polylog(J/\epsilon))$ gates.
Hence, the gate complexity is
\begin{align}
    G=\cO( Q+\log R+\polylog(J/\epsilon))=\cO(Q+\polylog(\kappa\|{\bf C}\|/\epsilon))
    =
    \cO(\kappa \; \polylog(\kappa\|{\bf C}\|/\epsilon)\;.
\end{align}
The term $Q$ is due to the gates in QSP.
        %%%%%%

%%%%%%%%
\subsection{The case $\bB=0$}

This case is similar to Ref.~\cite{CKS17}. Note that we aim at solving 
\begin{align}
    \bA \bX = \bC \;,
\end{align}
and we can assume $\bA$ to be Hermitian without loss of generality since replacing 
\begin{align}
    \bA \mapsto \begin{pmatrix}
        {\bf 0} & \bA \cr \bA^\dagger & {\bf 0}  
    \end{pmatrix} \; , \; \bX \mapsto \begin{pmatrix}
        {\bf 0} & {\bf 0} \cr \bX & {\bf 0}  
    \end{pmatrix} \; , \; \bC \mapsto \begin{pmatrix}
        {\bC} & {\bf 0} \cr {\bf 0} & {\bf 0}  
    \end{pmatrix}
\end{align}
reduces the general case to the Hermitian $\bA$ case. In a sense, this matrix equation is solving a linear system $\bA \vec X_i=\vec C_i$ where $\bA$ is fixed, but the different instances correspond to the different columns of $\bC$.

It is convenient to consider an approximation to $1/\bA$ and applying it to $\bC$ rather than considering vectorization in this case, since it gives better error bounds without invoking Lemma~\ref{lem:frobeniuserror}.
We assume $\|\bA\|\le 1$, which implies $\|\bQ\|\le 1$, and the parameter $\kappa=\|\bQ^{-1}\|$
is related to the condition number. The following result is adapted and improved from Ref.~\cite{CKS17}.
\begin{lemma}
\label{lem:LCUB=0}
    Let $\bA = \bA^\dagger \in \mathbb C^{N \times N}$
    be invertible such that $\|\bA\|\le 1$ and  $\kappa=\|\bA^{-1}\|$, and $\varepsilon >0$ the error. Then, there exists positive constants $c_1$, $c_2$, $c_3$, and $c_4$ such that for $\delta_t = c_1 \varepsilon/\sqrt{\log(1/\varepsilon)} $, $\delta_\omega = c_2/(\kappa \sqrt{\log(1/\varepsilon)})$, $t_R = c_3 \kappa \sqrt{\log(1/\varepsilon)}$, and $\omega_J = c_4\sqrt{\log(1/\varepsilon)}$, the Hermitian matrix
    \begin{align}
      h(\bA):= \frac {\ri} {\sqrt{2\pi}}\delta_t \delta_\omega \sum_{r=0}^{R-1}   \sum_{j=-J}^J    \omega_j   e^{-\omega_j^2/2} e^{-\ri t_r \omega_j \bA} \;,
    \end{align}
    where $t_r:=r \delta_t$, $\omega_j:=j \delta_\omega$, $R:=\lceil t_R/\delta_t\rceil$, and $J:=\lceil \omega_J/\delta_\omega\rceil$,
    can be written as $h({\bf A})=(\one_N-\varepsilon({\bf A})) \frac 1 {\bf A}$ and
    \begin{align}
        \| \varepsilon({\bf A}) \| \le \varepsilon \;.
    \end{align}
\end{lemma}
\begin{proof}
    The proof follows similar steps to Lemma 11 of Ref.~\cite{CKS17}. Performing the sum over $r$ gives
    \begin{align}
        \bA h(\bA) = \frac {\ri} {\sqrt{2\pi}}\delta_t \delta_\omega     \sum_{j=-J}^J    \bA \omega_j   e^{-\omega_j^2/2} \frac{\one_N - e^{-\ri t_R \omega_j \bA}}{\one_N - e^{-\ri \delta_t \omega_j \bA}} \;.
    \end{align}
    With the proper choice of constants, we have $\|\delta_t \omega_j \bA\|\le \delta_t \omega_J = \cO(\varepsilon)$. A Taylor series expansion implies
    \begin{align}
     \frac 1  {\one_N - e^{-\ri \delta_t \omega_j \bA}}  = \frac 1 {\ri \delta_t \omega_j \bA} (\one_N + \ri \delta_t \omega_j \bA + \ldots )\;.
    \end{align}
    It follows that 
    \begin{align}
      \frac {\ri} {\sqrt{2\pi}}\delta_t \delta_\omega     \sum_{j=-J}^J    \bA \omega_j   e^{-\omega_j^2/2} \frac{1}{\one_N - e^{-\ri \delta_t \omega_j \bA}}& = 
       \frac {\ri} {\sqrt{2\pi}}  \delta_\omega     \sum_{j=-J}^J   e^{-\omega_j^2/2} + \cO(\varepsilon) \\ &= \one_N +\cO(e^{-(2\pi /\delta_\omega)^2/2})+
       \cO(e^{-\omega_J^2/2}) + \cO(\varepsilon) \\
       &=\one_{N}+\cO(\varepsilon)\;.  
    \end{align}
    The other relevant term is
    \begin{align}
       \frac {1} {\sqrt{2\pi}} \delta_\omega     \sum_{j=-J}^J        e^{-\omega_j^2/2} 
       e^{-\ri t_R \omega_j \bA} &=  \frac {1} {\sqrt{2\pi}} \delta_\omega     \sum_{j=-\infty}^\infty        e^{-\omega_j^2/2} 
       e^{-\ri t_R \omega_j \bA}+ \cO(e^{-\omega_J^2/2}) \\
       & = \sum_{k=-\infty}^\infty  e^{-(t_R  \bA + 2 \pi k /\delta_\omega)^2/2}+ \cO(e^{-\omega_J^2/2}) \\
       & =e^{-(t_R  \bA )^2/2}+\sum_{k\ne0}  e^{-(t_R  \bA + 2 \pi k /\delta_\omega)^2/2}+ \cO(e^{-\omega_J^2/2}) \;,
    \end{align}
    where we used Eq.~\eqref{eq:Diraccomb} and the Fourier transform of the Gaussian.
    Here we need the assumption $1/\delta_\omega = \Omega(t_R)$,  which is readily implied by the hypothesis with a proper choice of constants; for example we can choose $2 \pi /\delta_\omega \ge 2 t_R$.
    Note that $\cO(e^{-\omega_J^2/2})=\cO(\varepsilon)$ and
    $\|e^{-(t_R  \bA )^2/2}\| \le e^{-(t_R  /\kappa )^2/2}=\cO(\varepsilon)$. Also, $e^{-( \pi /\delta_\omega)^2/2}=\cO(\varepsilon)$ implying that the second term (the one with the sum over $k \ne 0$) is also within the same error order.

    In summary, we showed
    \begin{align}
        \bA h(\bA) = \one_N - \varepsilon(\bA) = \one_N + \cO(\varepsilon)
    \end{align}
    implying $\|\varepsilon(\bA)\|\le \varepsilon$ with the right choice of constants.
\end{proof}

The approximation $h(\bA)$ is already an LCU, i.e., a linear combination of Hamiltonian evolutions under $\bA$.
Next we consider the L1-norm of this LCU.
\begin{lemma}
    The $L_1$-norm of $h(\bA)$ in Lemma~\ref{lem:LCUB=0} satisfies
    \begin{align}
        y:=\frac 1{\sqrt{2\pi}}\delta_t \delta_\omega \sum_{r=0}^{R-1} \sum_{j=-J}^J |\omega_j| e^{-\omega_j^2/2} = \cO(\kappa \sqrt{\log(1/\varepsilon)}) \;.
    \end{align}
\end{lemma}
\begin{proof}
    Note that 
    \begin{align}
    y= t_R \frac 1 {\sqrt{2\pi}}\delta_\omega\sum_{j=-J}^J |\omega_j| e^{-\omega_j^2/2} 
    \end{align}
    and we can upper bound this by an infinite sum. In addition, 
    \begin{align}
       \frac 1 {\sqrt{2\pi}}\delta_\omega\sum_{j=-\infty}^\infty |\omega_j| e^{-\omega_j^2/2}&  = 2 \frac 1 {\sqrt{2\pi}}\int_0^\infty \rd \omega \; \omega e^{-\omega^2/2} + \cO(\varepsilon) \\
       & = \frac 2 {\sqrt{2\pi}} + \cO(\varepsilon) \;.
    \end{align}
    Hence, $y \approx t_R \frac 2 {\sqrt{2\pi}}=\cO(\kappa\sqrt{\log(1/\varepsilon)})$.
\end{proof}

We are ready to obtain the query and gate complexities for this case. First note that, in contrast with other cases, we can set $\varepsilon=\cO(\epsilon)$. This is due to 
\begin{align}
  \|\bX'-\bX\| &=   \| (h(\bA) - \frac 1 {\bA})\bC\| \\
  & = \| \varepsilon(\bA) \bX \| \\
  & \le \varepsilon \|\bX\| \;.
\end{align}
Hence $\varepsilon$ is the multiplicative error and when we rescale by $x \ge \|\bX\|$ in the LCU, it becomes the additive error, so it suffices to set $\varepsilon=\cO(\epsilon)$. This gives $x:=\alpha y =\cO(\alpha \kappa \sqrt{\log(1/\epsilon)})$. The query complexity is then determined by the largest evolution time in the LCU. This is linear in $t_R \omega_J =\cO(\kappa \log(1/\varepsilon))=\cO(\kappa \log(1/\epsilon))$. We can use a method like QSP to simulate this evolution, and it suffices to do so within additive error $\cO(\epsilon)$ due to Lemma~\ref{lem:evolutionerror}. This gives the query complexity
\begin{align}
    Q= \cO(\kappa \log(1/\epsilon)+ \log(1/\epsilon))=
    \cO(\kappa \log(1/\epsilon)) \;.
\end{align}

The additional gate complexity is due to the unitary $V$ that prepares the ancillary register. Note that we can write
\begin{align}
    U_{\bX/x}=\frac 1 x \sum_i x_i V_i U_{\bC/\alpha}
\end{align}
where the $V_i's$ are time evolutions with $\bA$. 
We can assume $x_i \ge 0$ by absorbing the phase in the unitaries $V_i$. The unitary in the ancilla register is then
\begin{align}
    V \ket 0 \mapsto \frac 1 {\sqrt x} \sum_i \sqrt{x_i} \ket i \;.
\end{align}
  The registers $\ket i$ are composed of two registers, one to encode the times $t_r$ and the other to encode the frequencies $\omega_j$. The time register is in uniform superposition and can be prepares with $\log_2 R$ Hadamard gates. The amplitudes for the frequency register are according to the first Hermite function, and we can use the results in Ref.~\cite{Som15} to prepare this state with complexity $\cO(\log(J/\epsilon))$. Hence, the gate complexity is
\begin{align}
    G=\cO(Q+ \log(RJ/\epsilon)) =\cO(Q+\log(\kappa/\epsilon))=\cO(\kappa \log(1/\epsilon))\;.
\end{align}
The term $Q$ is due to the gates for QSP.

%%%%%%%%%%
\subsection{Positive matrices}

This is a special case where we assume access to the `square roots' of $\bA$ and $\bB$. 
For systems of linear equations, this case is analyzed in Ref.~\cite{CS16}, where an identity based on Hubbard-Stratonovich is used. Note that $\bQ \succ 0$ and, as before, we assume $\|\bQ\|\le 1$ and $\kappa=\|\bQ^{-1}\|$ is related to the condition number. The following result improves upon the results of Ref.~\cite{CS16}.

\begin{lemma}[Discrete time approximation to the inverse for positive matrices]
\label{lem:LCUpositive}
Let $\bQ = (\bP_\bA)^2 \otimes \one_N + \one_N \otimes (\bP_\bB)^2$ be positive and invertible with $\|\bQ\|\le 1$, $\kappa =\|\bQ^{-1}\|>0$, and $\varepsilon>0$ the error. Then, there exists positive constants $c_1$, $c_2$, $c_3$, and $c_4$ such that, for $\delta_t = c_1 \varepsilon$, $\delta_\omega=c_2/\sqrt{\kappa \log(1/\varepsilon)}$, $t_R=c_3 \kappa \log(1/\varepsilon)$, and $\omega_J=c_4 \sqrt{\log(\kappa/\varepsilon) }$, the matrix 
    \begin{align}
    h(\bQ):=\frac 1 {2\pi} \delta_t (\delta_\omega)^2 \sum_{r=0}^{R-1}
    \sum_{j=-J}^J  \sum_{j'=-J}^J e^{-(\omega_j^2 + \omega_{j'}^2)/2} e^{-\ri \omega_j \sqrt{2t_r}\bP_\bA } \otimes e^{-\ri \omega_{j'} \sqrt{2t_r}\bP_\bB }
\end{align}
where $t_r:=r \delta_t$, $\omega_j:=j \delta_\omega$, $R :=\lceil t_R/\delta_t\rceil$, and $J :=\lceil \omega_J/\delta_\omega\rceil$, can be written as $h(\bQ)=(\one_{N^2}-\varepsilon(\bQ)) \frac 1 \bQ$, and
\begin{align}
    \|\varepsilon(\bQ) - \one_{N^2}\|\le \varepsilon \;.
\end{align}
\end{lemma}

\begin{proof}
    Our definition of $h(\bQ)$ is intended to approximate  the identity in 
    Eq.~\eqref{eq:case4solution} (in vectorized form).
    We start from a different approximation:
    \begin{align}
    h'(\bQ)&:=\delta_t \sum_{r=0}^{R-1} e^{-t_r \bQ}  = \delta_t   \frac{\one_{N^2}-e^{-t_{R} \bQ}}{\one_{N^2}-e^{-\delta_t \bQ}}\;.
\end{align}
Since $\|\delta_t \bQ\| \le \delta_t = \cO(\varepsilon)$ under the assumptions, we can perform a Taylor series expansion and write
\begin{align}
   \delta_t \bQ   \frac{\one_{N^2}}{\one_{N^2}-e^{-\delta_t \bQ}} &=  \bQ  \delta_t \frac 1 { \bQ  \delta_t} ( \one_{N^2} + \delta_t \bQ/2 + \ldots) \\
   &=\one_{N^2} + \cO(\varepsilon) \;.
\end{align}
In addition, under the assumptions we have $\|e^{-t_{R}\bQ}\| \le e^{-t_{R}/\kappa}=\cO(\varepsilon)$, implying
\begin{align}
    h'(\bQ) =(\one_{N^2}-\varepsilon'(\bQ)) \frac 1 {\bQ}
\end{align}
with $\|\varepsilon'(\bQ)\| =\cO(\varepsilon)$.

Next we bound the error between $h'(\bQ)$ with $h(\bQ)$. 
Note that $h(\bQ)$ is obtained by approximating $ e^{-t_r \bQ}=  e^{-t_r (\bP_\bA)^2} \otimes  e^{-t_r (\bP_\bB)^2}$
using
\begin{align}
  e^{-t_r (\bP_\bA)^2} = \frac 1 {\sqrt {2\pi}} \int_{-\infty}^\infty \rd \omega \; e^{-\omega^2/2}e^{-\ri \omega \sqrt{2t_r} \bP_\bA} \approx \frac 1 {\sqrt {2\pi}} \delta_\omega \sum_{j=-J}^J e^{-\omega_j^2/2}e^{-\ri \omega_j \sqrt{2t_r} \bP_\bA}  \;,
\end{align}
and a similar approximation for $ e^{-t_r (\bP_\bB)^2}$.
Consider first the infinite sum, i.e., 
\begin{align}
\frac 1 {\sqrt{2\pi}} \delta_\omega  \sum_{j=-\infty}^\infty   e^{-\omega_j^2  /2} e^{-\ri \omega_j \sqrt{2t_r}\bP_\bA }      = e^{-t_r (\bP_\bA)^2} +
\sum_{k \ne 0} e^{-(\sqrt{2t_r}\bP_\bA+ 2 \pi k /\delta_\omega)^2/2} \;,
\end{align}
where we used Eq.~\eqref{eq:Diraccomb}. Assume now that $1/\delta_\omega = \Omega(\sqrt{t_R})$, for example, by choosing the constants so that $2 \pi/\delta_\omega \ge 2 \sqrt{2t_R}$, which would be satisfied under the assumptions. Then, by doing so, the last term (which includes the sum over $k\ne 0$) is $\cO(\varepsilon^\kappa)$.
In addition, if we perform the finite sum so that $|j|\le J$, the additional error is $\cO(e^{-\omega_J^2/2})=\cO((\varepsilon/\kappa)^2)$ under the assumptions. Hence, we showed
\begin{align}
    \frac 1 {\sqrt{2\pi}} \delta_\omega  \sum_{j=-J}^J   e^{-\omega_j^2  /2} e^{-\ri \omega_j \sqrt{2t_r}\bP_\bA }
    = e^{-t_r (\bP_\bA)^2} + \cO((\varepsilon/\kappa)^2) + \cO(\varepsilon^\kappa) = e^{-t_r (\bP_\bA)^2} + \cO((\varepsilon/\kappa)^2) \;,
\end{align}
asymptotically.

For the last step we note
\begin{align}
    \bQ h(\bQ) &= \bQ h'(\bQ) + \bQ (h(\bQ)-h'(\bQ)) \\
    & = \one_{N^2} + \cO(\varepsilon) + \bQ (h(\bQ)-h'(\bQ)) \;,
\end{align}
and hence we need to establish an upper for the last term. We have
\begin{align}
 h(\bQ)-h'(\bQ) &=    \delta_t  \sum_{r=0}^{R-1}  \left(\frac 1 {2\pi} (\delta_\omega)^2
    \sum_{j=-J}^J  \sum_{j'=-J}^J e^{-(\omega_j^2 + \omega_{j'}^2)/2} e^{-\ri \omega_j \sqrt{2t_r}\bP_\bA } \otimes e^{-\ri \omega_{j'} \sqrt{2t_r}\bP_\bB } - e^{-t_r \bQ}\right) \\
    & =  \delta_t  \sum_{r=0}^{R-1} \left( (e^{-t_r (\bP_\bA)^2}+\cO((\varepsilon/\kappa))^2) 
    (e^{-t_r (\bP_\bB)^2}+\cO((\varepsilon/\kappa)^2)) 
     - e^{-t_r \bQ}\right) \;,
\end{align}
and we can bound
\begin{align}
   \| h(\bQ)-h'(\bQ) \| =\cO( \delta_t R  (\varepsilon/\kappa)^2) = \cO(\kappa \sqrt{\log(1/\varepsilon)}\varepsilon^2/\kappa^2) = \cO(\varepsilon)
\end{align}
since $\|e^{-t_r (\bP_\bA)^2}\|\le 1$, $\|e^{-t_r (\bP_\bB)^2}\|\le 1$. Then, $\bQ h(\bQ) = \one_{N^2} + \cO(\varepsilon)$ as desired. It follows that we can choose the constants so that $\|\varepsilon(\bQ)\|\le \varepsilon$.
    
\end{proof}

Next we consider the $L_1$ norm of this LCU.

\begin{lemma}
    The $L_1$-norm of $h(\bQ)$ in Lemma~\ref{lem:LCUpositive}
    satisfies
    \begin{align}
        y:= \frac 1 {2\pi} \delta_t (\delta_\omega)^2 \sum_{r=0}^{R-1}
    \sum_{j=-J}^J  \sum_{j'=-J}^J e^{-(\omega_j^2 + \omega_{j'}^2)/2} =\cO(\kappa \log(1/\varepsilon)) \;.
    \end{align}
\end{lemma}
\begin{proof}
    Note that
    \begin{align}
        y =\frac 1 {2\pi} t_R (\delta_\omega)^2\sum_{j=-J}^J  \sum_{j'=-J}^J e^{-(\omega_j^2 + \omega_{j'}^2)/2} \;.
    \end{align}
    In addition we can use Eq.~\eqref{eq:discreteGaussiansum} to show,  for example,
    \begin{align}
     \frac 1 {\sqrt{2\pi}}   \delta_\omega \sum_{j=-\infty}^\infty
     e^{-\omega_j^2/2} = 1 + \cO(\varepsilon) \;.
    \end{align}
    Then, $y \approx t_R = \cO(\kappa \log(1/\varepsilon))$
    as stated.
\end{proof}

We are ready to determine the complexities.
Applying the inverse of vectorization, our approximate solution is
\begin{align}
    \bX \approx \bX':= \frac 1 {2\pi} \delta_t (\delta_\omega)^2 \sum_{r=0}^{R-1}
    \sum_{j=-J}^J  \sum_{j'=-J}^J e^{-(\omega_j^2 + \omega_{j'}^2)/2} e^{-\ri \omega_j \sqrt{2t_r}\bP_\bA } \bC e^{-\ri \omega_{j'} \sqrt{2t_r}\bP_\bB } \;.
\end{align}
If we replace $\bC$ by its block-encoding $U_{\bC/\alpha}$, then again this is an LCU involving Hamiltonian evolutions. Our quantum algorithm prepares the block-encoding $\bX'/x$, where $x:=\alpha y=\cO(\alpha \kappa \log(1/\varepsilon)$
for this case. We can invoke Lemma~\ref{lem:frobeniuserror} to set $\varepsilon=\cO(\epsilon/\sqrt N)$ for overall additive error $\cO(\epsilon)$ in the approximation $\bX'/x$. This implies $x =\cO(\alpha \kappa \log(N/\epsilon))$.

The query complexity is also given by the largest evolution time, which in this case is linear in $\sqrt{2t_R} \omega_J=\cO(\sqrt \kappa \log(\kappa /\varepsilon))$. Using QSP within error $\cO(\epsilon)$ due to Lemma~\ref{lem:evolutionerror}, the query complexity is
\begin{align}
    Q = \cO(\sqrt \kappa \log(\kappa /\varepsilon) + \log(1/\epsilon)) = \cO(\sqrt \kappa \log(\kappa N /\epsilon))\;. 
\end{align}

For the gate complexity we need to account for the gates to prepare the ancillary register. Note that we can write
\begin{align}
    U_{\bX/x}=\frac 1 x \sum_i x_i V_i U_{\bC/\alpha} W_i
\end{align}
for some unitaries $V_i$ and $W_i$ that are the time evolutions. We are interested in the operation
\begin{align}
    V\ket 0 \mapsto \frac 1 {\sqrt x}\sum_i \sqrt{x_i}\ket i
\end{align}
where the ancillary register is composed of three registers, one to encode the times $t_r$, another for the frequencies $\omega_j$, and another for the frequencies $\omega_{j'}$. The time register is in uniform superposition and can be prepared with $\log_2 R$ Hadamard gates. The frequency registers are in states with Gaussian-like amplitudes and can be prepared with $\cO(\log(J/\epsilon))$ gates within additive error $\cO(\epsilon)$. Hence, the overall gate complexity is
\begin{align}
    G= \cO(Q + \log (RJ/\epsilon))= \cO(Q + \log(\kappa N/\epsilon))=\cO(\sqrt \kappa \log(\kappa N/\epsilon))\;.
\end{align}
The term $Q$ is due to the gates for QSP.

%%%%%%%%%%%%%%%%%%%%%%%%%%%%%%%%%%%%%
\section{Oracle separations}
\label{app:separations}

We provide the proofs to the results on oracle separations in Sec.~\ref{sec:separations}.

%%%%%
\subsection{Proof of Thm.~\ref{thm:separation1}}

    The first part is trivial: prepare $\ket 0$, apply $\bC$, and measure in the computational basis. The decision problem is solved by observing whether the output is $\ket 0$ or a state orthogonal to it.

    For the second part we note that
    \begin{align}
        \dket C = \frac 1 {\sqrt N} \sum_{j,k \in {\cal V}} \ket{j,k}
    \end{align}
    where ${\cal V}$ is the set of pairs $(j,k)$ corresponding to the 1's in the permutation matrix $\bC$. The two cases are determined by whether $(0,0)$ is in ${\cal V}$ or not. 

   We can consider two distinct instances, one in which $(0,0)$ is in ${\cal V}$  and one in which 
   $(0,0)$ is not in ${\cal V}$, by permuting, for example, the first two columns of $\bC$. For these two instances we satisfy
   \begin{align}
        \| \dket C - \dket {C'} \| =\frac 4 {\sqrt N}\;.
    \end{align}
    These two states can then be prepared with two distinct black-box unitaries that satisfy
     \begin{align}
        \| U_{\bC} - U_{{\bf C'}} \|    =\cO( 1/\sqrt N)\;.
    \end{align}
     Using an adversarial argument, any quantum circuit 
    that distinguishes these instances and then solves the decision problem with high probability requires at least $\Omega(\sqrt N)$ uses to these unitaries, including their inverses and controlled versions.

%%%%%
\subsection{Proof of Thm.~\ref{thm:separation2}}

One implication is trivial: to construct $ V_\bC$ we first note that with on query to $U_{\bC/\alpha}$ we can implement the map
\begin{align}
    \ket{k,0} \ket 0 \mapsto \ket{k,k}\ket 0 \mapsto \frac 1 {\alpha}\ket{k,C_k} \ket 0 + \ket{k^\perp} \;,
    % \mapsto \frac 1 {\|\ket{C_k}\|}\ket{k,C_k} \;.
\end{align} 
where $\ket{k^\perp}$ is orthogonal to $\ket 0$.
Note that the norm of $\ket{C_k}$ is either $1/\sqrt 2$ or 1, so to obtain the normalized state $\ket{k,C_k}/\|\ket{C_k}\|$
the last step can be done with, for example, amplitude amplification. In particular, 
since $\alpha$ is known, amplitude amplification can be implemented exactly in this case. Then, we effectively constructed  
the map $\ket{k,0}\mapsto \ket{k,C_k}/\|\ket{C_k}\|$, or the unitary $V_{\bf C}$, with constant queries to $U_{\bC/\alpha}$ and the inverse as required by amplitude amplification.

For the other implication we consider  a reduction from search~\cite{Gro96}. In particular, we consider the set of $\bC$'s whose columns are 1-sparse, their first entry is $1/\sqrt 2$, and there is a `marked' column $\sigma \in \{0,\ldots,N-1\}$ and $\sigma \ne 0$ whose first row contains also an entry 
$1/\sqrt 2$. All other entries are 1 and located over the main diagonal. We can specify these instances using a function $f:\{0,\ldots,N-1\} \mapsto \{0,1\}$ such that $f(\sigma)=1$,
and $f(k)=0$ otherwise. The goal is to search for the unknown $\sigma$. 
This can be done using the (inverse) block-encoding for $\bC$,   since
\begin{align}
    U_\bC^{-1} \ket 0 = \frac 1 {\sqrt{2}} (\ket 0 + \ket {\sigma})\;.
\end{align}
Note that $\|\bC\|=1$ so that $\alpha=1$ here.
A measurement outputs $\sigma$ with probability $1/2$.

Now consider the unitary $V_\bC$ that performs the map $\ket{k,0} \mapsto \ket{k,C_k}/\|\ket{C_k}\|$. This unitary can be implemented with a single oracle call to $f$, since $\ket{C_k}/\|\ket{C_k}\|=\ket k$ when $k \ne \sigma$, and $\ket{C_\sigma}/\|\ket{C_\sigma}\|=\ket 0$.
It is known that finding $\sigma$ with high probability requires $\Omega(\sqrt N)$ calls to the oracle to $f$~\cite{BHMT02}. Hence, implementing $U_\bC$ (or its inverse), which would allow us to solve search with high probability, requires $\Omega(\sqrt N)$ calls to $V_\bC$, its inverse or controlled-version.

%%%%%%%%%%%%%%%%%%%%%%%%%%%%%%%%%%%%%%%%
\section{More general solution}\label{app:nonpos}
It is possible to generalize the solution to the case of matrices that need not be positive or normal, but the approach yields exponential complexity in the condition number $\kappa$, which is undesired.
The inverse function can be approximated using polynomials within error $\epsilon$ as in Ref.~\cite{CKS17}
\begin{equation}\label{eq:CKS}
    \frac 1{{\bf Q}} \approx 4\sum_{j=0}^{j_0} (-1)^j \left[ \frac{\sum_{i=j+1}^b \binom{2b}{b+i}}{2^{2b}} \right]T_{2j+1}({\bf Q})\, ,
\end{equation}
where $T_{2j+1}$ are Chebyshev polynomials,
\begin{equation}
    j_0 = \left\lceil \sqrt{b\log(4b/\epsilon)}\right\rceil \, , \qquad b = \lceil \kappa^2\log(\kappa/\epsilon) \rceil \, ,
\end{equation}
and $\kappa$ is the condition number of ${\bf Q}$.
Reference \cite{CKS17} considered Hermitian matrices, but Eq.~\eqref{eq:CKS} holds for more general ${\bf Q}$, if we interpret odd powers as alternating ${\bf Q}$ and its Hermitian conjugate.

In a linear-combination approach to approximating the solution of linear equations \cite{CKS17}, the solution is efficient because the coefficients in the square brackets in Eq.~\eqref{eq:CKS} are less than 1, and the Chebyshev polynomials can be generated using steps of oblivious amplitude amplification.
Here, the problem is more difficult, because we aim to provide the solution using block encodings of ${\bf A}$ and ${\bf B}$, and these are to be applied before and after ${\bf C}$.
In this case oblivious amplitude amplification does not provide Chebyshev polynomials jointly in ${\bf A}$ and ${\bf B}$.

It is possible to expand the Chebyshev polynomials and provide the solution to the Sylvester equation using products of ${\bf A}$, ${\bf A}^\dagger$, ${\bf B}$ and ${\bf B}^\dagger$.
For each power of ${\bf Q}$ it can be translated to operations on ${\bf C}$ as, for the first few cases
\begin{align}
    {\bf Q}^\dagger &\mapsto {\bf A}^\dagger {\bf C} + {\bf C}{\bf B}^\dagger \, , \\
    {\bf Q}{\bf Q}^\dagger &\mapsto {\bf A} [{\bf A}^\dagger {\bf C} + {\bf C}{\bf B}^\dagger] + [{\bf A}^\dagger {\bf C} + {\bf C}{\bf B}^\dagger]{\bf B} \, , \\
    {\bf Q}^\dagger{\bf Q}{\bf Q}^\dagger &\mapsto {\bf A}^\dagger\left\{{\bf A} [{\bf A}^\dagger {\bf C} + {\bf C}{\bf B}^\dagger] + [{\bf A}^\dagger {\bf C} + {\bf C}{\bf B}^\dagger]{\bf B}\right\} + \left\{{\bf A} [{\bf A}^\dagger {\bf C} + {\bf C}{\bf B}^\dagger] + [{\bf A}^\dagger {\bf C} + {\bf C}{\bf B}^\dagger]{\bf B}\right\}{\bf B}^\dagger \, .
\end{align}
In this way the description of the solution in terms of Chebyshev polynomials in ${\bf Q}$ can be translated into a sum of powers of ${\bf Q}$, then each power of ${\bf Q}$ can be translated into a form that can be block encoded.

The problem with this approach is that it gives a block-encoding with a complexity that increases exponentially.
Expanding the Chebyshev polynomials into a sum of powers gives
\begin{equation}
    \sum_{k=0}^{j_0} \xi_k {\bf Q}^\dagger({\bf Q} {\bf Q}^\dagger)^{k} \, ,
\end{equation}
where
\begin{equation}
    \xi_k = 4\sum_{j=k}^{j_0} (-1)^{j+k} \frac{\sum_{i=j+1}^b \binom{2b}{b+i}}{2^{2b}}T_{2j+1,2k+1} \, ,
\end{equation}
with $T_{2j+1,2k+1}$ the coefficient of $x^{2k+1}$ for the Chebyshev polynomial $T_{2j+1}(x)$.
The important feature of the coefficients of the Chebyshev polynomials is that
\begin{equation}
    \sum_{k=0}^j |T_{2j+1,2k+1}| = 2^{2j+1}-1 \, .
\end{equation}
This means that we can only guarantee the upper bound
\begin{equation}
    \sum_{k=0}^{j_0} |\xi_k| \le 4 \sum_{j=0}^{j_0} (2^{2j+1}-1) < \frac 83 2^{2j_0+2} \, .
\end{equation}
There is some cancellation which is not accounted for in the upper bound, but numerically it is found that it is still exponential.

Solution of the Sylvester equation by block-encoding gives complexity proportional to the sum of absolute values of coefficients.
Because $j_0\in \widetilde{\mathcal{O}}(\kappa)$, the complexity is exponential in $\kappa$.
Similarly, since $\log(1/\epsilon)$ appears in the exponential, the logarithmic scaling of the complexity with $\epsilon$ is lost as well.
Although this method of solution has poor complexity scaling with $\kappa$ and $\epsilon$, it is still logarithmic in the dimension of the matrices, and does not make any assumption on positivity or the matrices being normal.

\end{document}